\documentclass[cleveref,thm-restate]{lipics-v2021}
\hideLIPIcs
\nolinenumbers

\usepackage{newunicodechar}
\usepackage{silence}
\newunicodechar{♢}{\tikz \node[inner sep=1.5,draw,diamond] {};}
\newunicodechar{☆}{\tikz \node[inner sep=1,draw,star,star point ratio=2] {};}
\newunicodechar{△}{\triangle}
\newunicodechar{⬜}{\kern 0.5pt\tikz \node[inner sep=1.7,draw,regular polygon,regular polygon sides=4] {};\kern 0.5pt}
\newunicodechar{◯}{\tikz[baseline=-3pt] \node[inner sep=1.7,draw,cloud,cloud puffs=4,cloud puff arc=190] {};}

\newunicodechar{⊥}{\bot}
\newunicodechar{•}{\item}
\newunicodechar{✓}{\checkmark}
\newunicodechar{✗}{\xmark}\def\xmark{\ding{55}}%
\newunicodechar{…}{\dots}
\newunicodechar{≔}{\coloneqq}
\newunicodechar{⁻}{^-}
\newunicodechar{⁺}{^+}
\newunicodechar{₋}{_-}
\newunicodechar{₊}{_+}
\newunicodechar{ℓ}{\ell}
\newunicodechar{•}{\item}
\newunicodechar{…}{\dots}
\newunicodechar{≔}{\coloneqq}
\newif\ifnormopen\normopenfalse
\newunicodechar{‖}{\ifnormopen\rVert\normopenfalse\else\rVert\normopentrue\fi}
\newunicodechar{≤}{\leq}
\newunicodechar{≥}{\geq}
\newunicodechar{≰}{\nleq}
\newunicodechar{≱}{\ngeq}
\newunicodechar{⊕}{\oplus}
\newunicodechar{⊗}{\otimes}
\newunicodechar{≠}{\neq}
\newunicodechar{¬}{\neg}
\newunicodechar{≡}{\equiv}
\newunicodechar{₀}{_0}
\newunicodechar{₁}{_1}
\newunicodechar{₂}{_2}
\newunicodechar{₃}{_3}
\newunicodechar{₄}{_4}
\newunicodechar{₅}{_5}
\newunicodechar{₆}{_6}
\newunicodechar{₇}{_7}
\newunicodechar{₈}{_8}
\newunicodechar{₉}{_9}
\newunicodechar{ₚ}{_p}
\newunicodechar{ₙ}{_n}
\newunicodechar{ₐ}{_a}
\newunicodechar{ₑ}{_e}
\newunicodechar{ₕ}{_h}
\newunicodechar{ₖ}{_k}
\newunicodechar{ₗ}{_l}
\newunicodechar{ₘ}{_m}
\newunicodechar{ₛ}{_s}
\newunicodechar{ₜ}{_t}
\newunicodechar{ₓ}{_x}
\newunicodechar{⁰}{^0}
\newunicodechar{¹}{^1}
\newunicodechar{²}{^2}
\newunicodechar{³}{^3}
\newunicodechar{⁴}{^4}
\newunicodechar{⁵}{^5}
\newunicodechar{⁶}{^6}
\newunicodechar{⁷}{^7}
\newunicodechar{⁸}{^8}
\newunicodechar{⁹}{^9}
\newunicodechar{ⁿ}{^n}
\newunicodechar{∈}{\in}
\newunicodechar{∉}{\notin}
\newunicodechar{⊂}{\subset}
\newunicodechar{⊃}{\supset}
\newunicodechar{⊆}{\subseteq}
\newunicodechar{⊇}{\supseteq}
\newunicodechar{⊄}{\nsubset}
\newunicodechar{⊅}{\nsupset}
\newunicodechar{⊈}{\nsubseteq}
\newunicodechar{⊉}{\nsupseteq}
\newunicodechar{∪}{\cup}
\newunicodechar{∩}{\cap}
\newunicodechar{∀}{\forall}
\newunicodechar{∃}{\exists}
\newunicodechar{∄}{\nexists}
\newunicodechar{∨}{\vee}
\newunicodechar{∧}{\wedge}
\newunicodechar{⊼}{\bar{\wedge}}
\newunicodechar{⊽}{\bar{\vee}}
\newunicodechar{ℝ}{\mathbb{R}}
\newunicodechar{ℙ}{\mathbb{P}}
\newunicodechar{ℕ}{\mathbb{N}}
\newunicodechar{𝔼}{\mathbb{E}}
\newunicodechar{𝔽}{\mathbb{F}}
\newunicodechar{ℤ}{\mathbb{Z}}
\newunicodechar{⌊}{\lfloor}
\newunicodechar{⌋}{\rfloor}
\newunicodechar{⌈}{\lceil}
\newunicodechar{⌉}{\rceil}
\newunicodechar{·}{\cdot}
\newunicodechar{∘}{\circ}
\newunicodechar{×}{\times}
\newunicodechar{↑}{\uparrow}
\newunicodechar{↓}{\downarrow}
\newunicodechar{→}{\rightarrow}
\newunicodechar{←}{\leftarrow}
\newunicodechar{⇒}{\Rightarrow}
\newunicodechar{⇐}{\Leftarrow}
\newunicodechar{↔}{\leftrightarrow}
\newunicodechar{⇔}{\Leftrightarrow}
\newunicodechar{↦}{\mapsto}
\newunicodechar{∅}{\emptyset}
\newunicodechar{∞}{\infty}
\newunicodechar{≅}{\cong}
\newunicodechar{≈}{\approx}
\newunicodechar{ℓ}{\ell}
\newunicodechar{𝟙}{\mathds{1}}
\newunicodechar{𝟘}{\mathds{0}}
\newunicodechar{↪}{\hookrightarrow}

\newunicodechar{α}{\alpha}
\newunicodechar{β}{\beta}
\newunicodechar{γ}{\gamma}
\newunicodechar{Γ}{\Gamma}
\newunicodechar{δ}{\delta}
\newunicodechar{Δ}{\Delta}
\newunicodechar{ε}{\varepsilon}
\newunicodechar{ζ}{\zeta}
\newunicodechar{η}{\eta}
\newunicodechar{θ}{\theta}
\newunicodechar{Θ}{\Theta}
\newunicodechar{ι}{\iota}
\newunicodechar{κ}{\kappa}
\newunicodechar{λ}{\lambda}
\newunicodechar{Λ}{\Lambda}
\newunicodechar{μ}{\mu}
\newunicodechar{ν}{\nu}
\newunicodechar{ξ}{\xi}
\newunicodechar{Ξ}{\Xi}
\newunicodechar{π}{\pi}
\newunicodechar{Π}{\Pi}
\newunicodechar{ρ}{\rho}
\newunicodechar{σ}{\sigma}
\newunicodechar{Σ}{\Sigma}
\newunicodechar{τ}{\tau}
\newunicodechar{υ}{\upsilon}
\newunicodechar{ϒ}{\Upsilon}
\newunicodechar{φ}{\varphi}
\newunicodechar{ϕ}{\phi}
\newunicodechar{Φ}{\Phi}
\newunicodechar{χ}{\chi}
\newunicodechar{ψ}{\psi}
\newunicodechar{Ψ}{\Psi}
\newunicodechar{ω}{\omega}
\newunicodechar{Ω}{\Omega}

\newunicodechar{𝒜}{\mathcal{A}}
\newunicodechar{ℬ}{\mathcal{B}}
\newunicodechar{𝒞}{\mathcal{C}}
\newunicodechar{𝒟}{\mathcal{D}}
\newunicodechar{ℰ}{\mathcal{E}}
\newunicodechar{ℱ}{\mathcal{F}}
\newunicodechar{𝒢}{\mathcal{G}}
\newunicodechar{ℋ}{\mathcal{H}}
\newunicodechar{ℐ}{\mathcal{I}}
\newunicodechar{𝒥}{\mathcal{J}}
\newunicodechar{𝒦}{\mathcal{K}}
\newunicodechar{ℒ}{\mathcal{L}}
\newunicodechar{ℳ}{\mathcal{M}}
\newunicodechar{𝒩}{\mathcal{N}}
\newunicodechar{𝒪}{\mathcal{O}}
\newunicodechar{𝒫}{\mathcal{P}}
\newunicodechar{𝒬}{\mathcal{Q}}
\newunicodechar{ℛ}{\mathcal{R}}
\newunicodechar{𝒮}{\mathcal{S}}
\newunicodechar{𝒯}{\mathcal{T}}
\newunicodechar{𝒰}{\mathcal{U}}
\newunicodechar{𝒱}{\mathcal{V}}
\newunicodechar{𝒲}{\mathcal{W}}
\newunicodechar{𝒳}{\mathcal{X}}
\newunicodechar{𝒴}{\mathcal{Y}}
\newunicodechar{𝒵}{\mathcal{Z}}
\newunicodechar{𝒶}{\mathcal{a}}
\newunicodechar{𝒷}{\mathcal{b}}
\newunicodechar{𝒸}{\mathcal{c}}
\newunicodechar{𝒹}{\mathcal{d}}
\newunicodechar{ℯ}{\mathcal{e}}
\newunicodechar{𝒻}{\mathcal{f}}
\newunicodechar{ℊ}{\mathcal{g}}
\newunicodechar{𝒽}{\mathcal{h}}
\newunicodechar{𝒾}{\mathcal{i}}
\newunicodechar{𝒿}{\mathcal{j}}
\newunicodechar{𝓀}{\mathcal{k}}
\newunicodechar{𝓁}{\mathcal{l}}
\newunicodechar{𝓂}{\mathcal{m}}
\newunicodechar{𝓃}{\mathcal{n}}
\newunicodechar{ℴ}{\mathcal{o}}
\newunicodechar{𝓅}{\mathcal{p}}
\newunicodechar{𝓆}{\mathcal{q}}
\newunicodechar{𝓇}{\mathcal{r}}
\newunicodechar{𝓈}{\mathcal{s}}
\newunicodechar{𝓉}{\mathcal{t}}
\newunicodechar{𝓊}{\mathcal{u}}
\newunicodechar{𝓋}{\mathcal{v}}
\newunicodechar{𝓌}{\mathcal{w}}
\newunicodechar{𝓍}{\mathcal{x}}
\newunicodechar{𝓎}{\mathcal{y}}
\newunicodechar{𝓏}{\mathcal{z}}

\usepackage{enumerate}
\usepackage{booktabs}
\usepackage{amsmath,amssymb,amsthm,mathtools}
\usepackage{bm}
\usepackage{subcaption}
\usepackage{pgfplots}
\pgfplotsset{compat=newest}

\usepgfplotslibrary{groupplots}
\usetikzlibrary{pgfplots.polar}
\usetikzlibrary{pgfplots.statistics}
\usetikzlibrary{arrows.meta}
\usetikzlibrary{overlay-beamer-styles}
\pgfplotsset{every axis/.style={scale only axis}}

\definecolor{veryLightGrey}{HTML}{F2F2F2}
\definecolor{lightGrey}{HTML}{DDDDDD}
\definecolor{colorSimdRecSplit}{HTML}{444444}
\definecolor{colorChd}{HTML}{377EB8}
\definecolor{colorRustFmph}{HTML}{A65628}
\definecolor{colorRustFmphGo}{HTML}{A65628}
\definecolor{colorSicHash}{HTML}{4DAF4A}
\definecolor{colorPthash}{HTML}{984EA3}
\definecolor{colorDensePtHash}{HTML}{377EB8}
\definecolor{colorRecSplit}{HTML}{FF7F00}
\definecolor{colorBbhash}{HTML}{F781BF}
\definecolor{colorShockHash}{HTML}{F8BA01}
\definecolor{colorBipartiteShockHash}{HTML}{F781BF}
\definecolor{colorBipartiteShockHashFlat}{HTML}{E41A1C}
\definecolor{colorBmz}{HTML}{000000}
\definecolor{colorBdz}{HTML}{444444}
\definecolor{colorFch}{HTML}{444444}
\definecolor{colorChm}{HTML}{A65628}
\definecolor{colorFiPS}{HTML}{FF7F00}
\definecolor{colorConsensus}{HTML}{4DAF4A}

\newlength{\squareWidth}
\newlength{\squareHeight}
\pgfdeclareplotmark{heatmapSquare}{%
  \pgfpathmoveto{\pgfqpoint{ \squareWidth}{ \squareHeight}}%
  \pgfpathlineto{\pgfqpoint{ \squareWidth}{-\squareHeight}}%
  \pgfpathlineto{\pgfqpoint{-\squareWidth}{-\squareHeight}}%
  \pgfpathlineto{\pgfqpoint{-\squareWidth}{ \squareHeight}}%
  \pgfpathclose%
  \pgfusepath{fill}
}

\pgfdeclareimage[interpolate,height=2.0mm,width=2.0mm]{shockhashMark}{fig/shockhash}
\pgfdeclareplotmark{shockhash}{\pgftext[at=\pgfpointorigin]{\pgfuseimage{shockhashMark}}}
\pgfdeclareimage[interpolate,height=1.85mm,width=1.85mm]{shockhashInverseMark}{fig/shockhashInverse}
\pgfdeclareplotmark{shockhashInverse}{\pgftext[at=\pgfpointorigin]{\pgfuseimage{shockhashInverseMark}}}
\pgfdeclareimage[interpolate,height=1.85mm,width=1.85mm]{phobicMark}{fig/phobic}
\pgfdeclareplotmark{phobic}{\pgftext[at=\pgfpointorigin]{\pgfuseimage{phobicMark}}}
\pgfdeclareimage[interpolate,height=1.85mm,width=1.85mm]{phobicInverseMark}{fig/phobicInverse}
\pgfdeclareplotmark{phobicInverse}{\pgftext[at=\pgfpointorigin]{\pgfuseimage{phobicInverseMark}}}
\pgfdeclareimage[interpolate,height=1.85mm,width=1.85mm]{fingerprintMark}{fig/hand}
\pgfdeclareplotmark{fingerprint}{\pgftext[at=\pgfpointorigin]{\pgfuseimage{fingerprintMark}}}
\pgfdeclareimage[interpolate,height=1.85mm,width=1.85mm]{fingerprintInverseMark}{fig/handInverse}
\pgfdeclareplotmark{fingerprintInverse}{\pgftext[at=\pgfpointorigin]{\pgfuseimage{fingerprintInverseMark}}}
\pgfdeclareimage[interpolate,height=2.0mm,width=2.0mm]{consensusMark}{fig/consensus}
\pgfdeclareplotmark{consensus}{\pgftext[at=\pgfpointorigin]{\pgfuseimage{consensusMark}}}

\pgfplotscreateplotcyclelist{myColorList}{%
  colorShockHash,mark=triangle\\
  lipicsGray,mark=pentagon\\%
  colorChd,mark=square\\%
  colorSicHash,mark=o\\%
  colorRecSplit,mark=diamond\\%
  colorPthash,mark=x\\%
  colorBbhash,mark=flippedTriangle\\%
}
\pgfdeclareplotmark{flippedTriangle}{%
  \pgfpathmoveto{\pgfqpointpolar{-90}{1.2\pgfplotmarksize}}%
  \pgfpathlineto{\pgfqpointpolar{30}{1.2\pgfplotmarksize}}%
  \pgfpathlineto{\pgfqpointpolar{150}{1.2\pgfplotmarksize}}%
  \pgfpathclose%
  \pgfusepath{stroke}
}
\pgfdeclareplotmark{rightTriangle}{%
  \pgfpathmoveto{\pgfqpointpolar{0}{1.2\pgfplotmarksize}}%
  \pgfpathlineto{\pgfqpointpolar{120}{1.2\pgfplotmarksize}}%
  \pgfpathlineto{\pgfqpointpolar{240}{1.2\pgfplotmarksize}}%
  \pgfpathclose%
  \pgfusepath{stroke}%
}
\pgfdeclareplotmark{leftTriangle}{%
  \pgfpathmoveto{\pgfqpointpolar{-180}{1.2\pgfplotmarksize}}%
  \pgfpathlineto{\pgfqpointpolar{-60}{1.2\pgfplotmarksize}}%
  \pgfpathlineto{\pgfqpointpolar{60}{1.2\pgfplotmarksize}}%
  \pgfpathclose%
  \pgfusepath{stroke}%
}

\definecolor{veryLightGrey}{HTML}{DDDDDD}

\pgfplotsset{
  mark repeat*/.style={
    scatter,
    scatter src=x,
    scatter/@pre marker code/.code={
      \pgfmathtruncatemacro\usemark{
        or(mod(\coordindex,#1)==0, (\coordindex==(\numcoords-1))
      }
      \ifnum\usemark=0
        \pgfplotsset{mark=none}
      \fi
    },
    scatter/@post marker code/.code={}
  },
  log x ticks with fixed point/.style={
      xticklabel={
        \pgfkeys{/pgf/fpu=true}
        \pgfmathparse{exp(\tick)}%
        \pgfmathprintnumber[fixed relative, precision=3]{\pgfmathresult}
        \pgfkeys{/pgf/fpu=false}
      }
  },
  log y ticks with fixed point/.style={
      yticklabel={
        \pgfkeys{/pgf/fpu=true}
        \pgfmathparse{exp(\tick)}%
        \pgfmathprintnumber[fixed relative, precision=3]{\pgfmathresult}
        \pgfkeys{/pgf/fpu=false}
      }
  },
  major grid style={thin,dotted},
  minor grid style={thin,dotted},
  ymajorgrids,
  yminorgrids,
  every axis/.append style={
    line width=0.9pt,
    tick style={
      line cap=round,
      thin,
      major tick length=4pt,
      minor tick length=2pt,
    },
    mark options={solid},
  },
  legend cell align=left,
  legend style={
    line width=0.7pt,
    /tikz/every even column/.append style={column sep=2mm,black},
    /tikz/every odd column/.append style={black},
    mark options={solid},
    font=\footnotesize,
  },
  title style={yshift=-2pt},
  enlarge x limits=0.04,
  every tick label/.append style={font=\footnotesize},
  every axis label/.append style={font=\small},
  every axis y label/.append style={yshift=-1ex},
  /pgf/number format/1000 sep={},
  axis lines*=left,
  xlabel near ticks,
  ylabel near ticks,
  axis lines*=left,
  label style={font=\footnotesize},
  tick label style={font=\footnotesize},
  cycle list name=myColorList,
}

\newif\iffullversion
\fullversiontrue

\usepackage[commentsnumbered,ruled,vlined]{algorithm2e} %

    \SetCommentSty{commentFont}
    \SetKwProg{algo}{Algorithm}{:}{}
    \usepackage{etoolbox}
    \patchcmd{\SetProgSty}{ArgSty}{ProgSty}{}{}
    \SetProgSty{}
    \SetKw{And}{and}
    \SetKw{Or}{or}
    \DontPrintSemicolon
\def\Ber{\mathrm{Ber}}
\def\Geom{\mathrm{Geom}}

\def\OPT{\mathrm{OPT}}
\def\Smin{\vec{S}_{\mathrm{MIN}}}
\def\OPTMPHF{\OPT_{\mathrm{MPHF}}}
\def\consensus{\texorpdfstring{C\scalebox{0.8}{ONSENSUS}}{CONSENSUS}\xspace}
\def\textrel#1#2{\stackrel{\text{#1}}{#2}}
\def\vecB{\vec{B}}
\def\B#1{B^{(#1)}}
\def\vB#1{\vecB^{(#1)}}

\newcommand{\psfrage}[1]{{\textcolor{blue}{[ps:#1]}}}
\newcommand{\hpfrage}[1]{{\color{violet}\sf[HP: #1]}}
\renewcommand{\psfrage}[1]{} \renewcommand{\hpfrage}[1]{}

\title{Combined Search and Encoding for Seeds,\texorpdfstring{\\ }{ }with an Application to Minimal Perfect Hashing}
\titlerunning{Combined Search and Encoding for Seeds}

\author{Hans-Peter Lehmann}{Karlsruhe Institute of Technology, Germany}{hans-peter.lehmann@kit.edu}{https://orcid.org/0000-0002-0474-1805}{}
\author{Peter Sanders}{Karlsruhe Institute of Technology, Germany}{sanders@kit.edu}{https://orcid.org/0000-0003-3330-9349}{}
\author{Stefan Walzer}{Karlsruhe Institute of Technology, Germany}{stefan.walzer@kit.edu}{https://orcid.org/0000-0002-6477-0106}{}
\author{Jonatan Ziegler}{Karlsruhe Institute of Technology, Germany}{jonatan.ziegler@student.kit.edu}{}{}
\authorrunning{Lehmann, Sanders, Walzer, Ziegler}
\Copyright{Hans-Peter Lehmann, Peter Sanders, Stefan Walzer, Jonatan Ziegler}


\funding{
\flag[2cm]{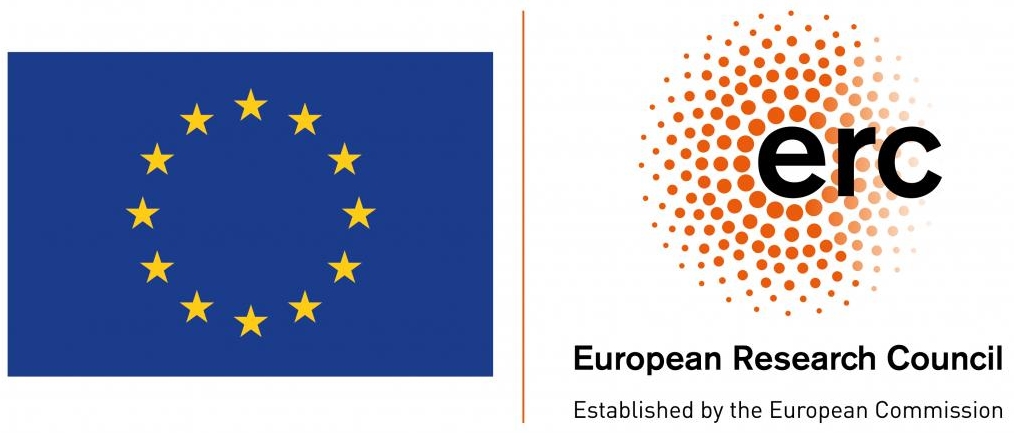}
    This project has received funding from the European Research Council (ERC) under the European Union’s Horizon 2020 research and innovation programme (grant agreement No. 882500).
    This work was also supported by funding from the pilot program Core Informatics at KIT (KiKIT) of the Helmholtz Association (HGF).
}
\acknowledgements{
    This paper is based on the bachelor's thesis of Jonatan Ziegler \cite{ziegler2025compacting}.
}

\supplementdetails[subcategory={Source code}]{Software}{https://github.com/ByteHamster/ConsensusRecSplit}
\supplementdetails[subcategory={Comparison with competitors}]{Software}{https://github.com/ByteHamster/MPHF-Experiments}

\EventEditors{Anne Benoit, Haim Kaplan, Sebastian Wild, and Grzegorz Herman}
\EventNoEds{4}
\EventLongTitle{33rd Annual European Symposium on Algorithms (ESA 2025)}
\EventShortTitle{ESA 2025}
\EventAcronym{ESA}
\EventYear{2025}
\EventDate{September 15--17, 2025}
\EventLocation{Warsaw, Poland}
\EventLogo{}
\SeriesVolume{351}
\ArticleNo{108}

\begin{document}

\maketitle
\begin{abstract}
    Randomised algorithms often employ methods that can fail and that are retried with independent randomness until they succeed. Randomised data structures therefore often store indices of successful attempts, called seeds.
    If $n$ such seeds are required (e.g., for independent substructures) the standard approach is to compute for each $i ∈ [n]$ the smallest successful seed $S_i$ and store $\vec{S} = (S₁,…,Sₙ)$.

    The central observation of this paper is that this is not space-optimal. We present a different algorithm that computes a sequence $\vec{S}' = (S₁',…,Sₙ')$ of successful seeds such that the entropy of $\vec{S'}$ undercuts the entropy of $\vec{S}$ by $Ω(n)$ bits in most cases. To achieve a memory consumption of $\OPT+εn$, the expected number of inspected seeds increases by a factor of $𝒪(1/ε)$.
    
    We demonstrate the usefulness of our findings with a novel construction for minimal perfect hash functions that, for $n$ keys and any $ε ∈ [n^{-3/7},1]$, has space requirement $(1+ε)\OPT$ and construction time $𝒪(n/ε)$. All previous approaches only support $ε = ω(1/\log n)$ or have construction times that increase exponentially with $1/ε$.
    Our implementation beats the construction throughput of the state of the art by more than two orders of magnitude for $ε ≤ 3\%$.
    
\end{abstract}

\newpage
\section{Introduction}
The construction of randomised data structures can sometimes fail, so the data structures need to store indices of successful attempts called \emph{seeds}.%
\footnote{Throughout this paper, we use the \emph{Simple Uniform Hashing Assumption}, as do many works before us \cite{dietzfelbinger1990new,pagh2007linear,pagh2008uniform,lehmann2023shockhash,lehmann2024shockhash2,hermann2024phobic,esposito2020recsplit}. This is equivalent to saying that we assume access to an infinite sequence of uniformly random and independent hash functions, each identified by its index (seed). The assumption is an adequate model for practical hash functions and can be justified in theory using the “split-and-share” trick \cite{dietzfelbinger2009applications}.}
In this paper, we present a technique that finds and encodes successful seeds in a space-efficient and time-efficient way.

When searching a single seed, we may imagine an infinite sequence of i.i.d.\ Bernoulli random variables, called a \emph{Bernoulli process}, that indicate which natural numbers lead to success and that can be inspected one by one.
To model the task of finding a \emph{sequence} of seeds we use a \emph{sequence} of Bernoulli processes as follows. In this paper $[n] \coloneq \{1,…,n\}$ and $ℕ₀ \coloneq ℕ\cup\{0\}$.

\begin{definition}[Seed Search and Encoding Problem (SSEP)]
    Let $n ∈ ℕ$ and $p₁,…,pₙ ∈ (0,1]$. Given for each $i ∈ [n]$ a Bernoulli process $\vB{i} = (\B{i}_j)_{j∈ℕ₀}$ with parameter $p_i$, the task is to craft a bitstring $M$ that encodes for each $i ∈ [n]$ a successful seed $S_i ∈ ℕ₀$, i.e.\ $\B{i}_{S_i} = 1$.   
\end{definition}
Our primary goals are to minimise the length of $M$ in bits, and to minimise for each $i ∈ [n]$ the number $T_i$ inspected Bernoulli random variables from $\vB{i}$. Moreover, the time to decode $S_i$ given $p₁,…,pₙ$, $M$ and $i ∈ [n]$ should be constant.
The optimal values for the expectation of $|M|$ and $T_i$ are the following
\iffullversion
(proved in \cref{sec:lower-bounds-analysis}).
\else
(for a proof see the full version of this paper \cite{lehmann2025combinedArxiv}).
\fi
\begin{equation}
    M^\OPT = \sum_{i = 1}^{n} \log₂(1/p_i) \text{ and } T_i^\OPT = 1/p_i \text{ for $i ∈ [n]$}.
    \label{eq:lower-bounds}
\end{equation}

\subsection{Ad-Hoc Solutions}\label{s:entropyDifference}
We briefly consider two natural strategies for solving the SSEP to build intuition.
The obvious strategy (MIN) is to encode the smallest successful seeds, i.e.\ define
\[ S_i \coloneq \min \{j ∈ ℕ₀ \mid \B{i}_j = 1\} \text{ and } \Smin \coloneq (S₁,…,Sₙ) \tag{MIN}\]
and let $M$ be an optimal encoding of $\Smin$. Perhaps surprisingly, this requires $Ω(n)$ bits more than the optimum if $p₁,…,pₙ$ are bounded away from $1$.
\def\seeApp{%
\iffullversion%
, proof in \cref{s:adhoc-proofs}%
\else%
, proof in full version \cite{lehmann2025combinedArxiv}%
\fi}
\begin{restatable}[MIN is fast but not space-efficient\seeApp]{lemma}{MINperf}\ \\
    \label{lem:ad-hoc:min}
    If $M$ encodes $\Smin$ then $𝔼[T_i] = T_i^\OPT$ for all $i ∈ [n]$ but $𝔼[|M|] ≥ M^\OPT + \sum_{i = 1}^{n} 1-p_i$.
\end{restatable}

Another simple strategy (UNI) is to encode the smallest seed that is simultaneously successful for all Bernoulli processes, i.e.\ we define
\[ S_* = \min \{ j ∈ ℕ₀ \mid ∀i ∈ [n]: \B{i}_j = 1\}
\text{ and }
S_i = S_* \text{ for all $i ∈ [n]$}. \tag{UNI} \]
This yields a bitstring $M$ of minimal length but requires total time $\sum_{i = 1}^{n} T_i = \exp(Ω(n))$ if $p₁,…,pₙ$ are bounded away from $1$.
\begin{restatable}[UNI is space-efficient but slow\seeApp]{lemma}{UNIperf}\ \\ 
    \label{lem:ad-hoc:uni}
    If $M$ encodes $S_*$ then $𝔼[|M|] = M^\OPT + 𝒪(1)$ but $𝔼[\sum_{i = 1}^{n} T_i] ≥ \prod_{i = 1}^{n} 1/p_i \gg \sum_{i = 1}^{n} T_i^\OPT$.
\end{restatable}

\subsection{The \consensus Algorithm}
We present an algorithm for the SSEP designed to simultaneously get close to both optima. We call it \emph{Combined Search and Encoding for Successful Seeds}, or \consensus for short. We summarise the guarantees of the stronger of two variants here.
\begin{theorem}[Main result, informal version of \cref{thm:full-consensus}]
    \label{thm:consise}
    For any $ε > 0$ the full \consensus algorithm with high probability solves the SSEP with $|M| ≤ M^\OPT+εn+𝒪(\log n)$ and $𝔼[T_i] = 𝒪(T_i^\OPT / ε)$ for all $i ∈ [n]$. Each seed $S_i$ is a bitstring of fixed length found at a predetermined offset in $M$.
\end{theorem}
Note that if the geometric mean of $p₁,…,pₙ$ is at most $\frac 12$ then $M^\OPT ≥ n$ and the leading term $εn$ in our space overhead is at most $ε M^\OPT$.%
\footnote{It is possible to get $𝔼[|M|] ≤ M^\OPT+ε\min(M^\OPT,n)+𝒪(\log n)$ and $𝔼[T_i] = 𝒪(T_i^\OPT/ε)$ in general. This requires work if many $p_i$ are close to $1$ such that $M^\OPT = o(n)$. In that case we can “summarise” subsequences of $p₁,…,pₙ$ using essentially the UNI strategy and then apply our main theorem to a shorter sequence $\tilde{p}₁,…,\tilde{p}_{n'}$ with $n' ≤ M^\OPT + 𝒪(1)$. We decided against integrating this improvement into our arguments as it interferes with clarity.}

\subparagraph{Intuition.} First reconsider the MIN-strategy. It involves the sequence $\Smin = (S₁,…,Sₙ)$ with some entropy $H(\Smin)$. There are two issues with storing $\Smin$. The superficial issue is that each number in the sequence follows a geometric distribution, and it is not obvious how to encode $\Smin$ in space close to $H(\Smin)$ while preserving fast random access to each component.\footnote{The standard solution \cite{witten1999managing} using Golomb-Rice coding \cite{golomb1966run,rice1979some} comes with $Ω(n)$ bits of overhead.}
The deeper issue is that $\Smin$ is not what a space-optimal approach should encode in the first place.
This is demonstrated by the fact that even an optimal encoding of $\Smin$ needs space $H(\Smin)$, which UNI undercuts by $Ω(n)$ bits.

\consensus deals with both issues simultaneously. The idea is to store for each $i ∈ [n]$ a natural number $1 ≤ σ_i ≤ 2^{ε}/p_i$. This allows for fixed width encoding of $σ_i$ using $ε + \log₂ 1/p_i$ bits. If we simply used $S_i = σ_i$ then the expected number of successful choices for $σ_i$ would be $2^ε > 1$. The problem would be, of course, that the actual number of successful choices might be zero for many $i ∈ [n]$. So instead we interpret $σ_i$ as a \emph{seed fragment} and define $S_i$ to be a concatenation of seed fragments up to and including $σ_i$. This way, if we run out of choices for $σ_i$, we can hope that a different successful choice for earlier seed fragments will allow for a successful choice for $σ_i$ as well. We can construct the sequence $(σ₁,…,σₙ)$ using backtracking.
We explain our algorithm in detail in \cref{sec:algorithm}.
The main technical challenge is to bound the additional cost this brings.

\subsection{Motivation: Seed Sequences in Randomised Data Structures}
\label{sec:motivation}

Randomised data structures frequently use seeds to initialise pseudo-random number generators or as additional inputs to hash functions.
Since a seed normally counts how often a construction has been restarted, a seed $S ∈ ℕ₀$ implies that work $Ω(S+1)$ has been carried out, meaning any reasonably fast algorithm yields seeds that are only a few bits to a few bytes in size. This does not, however, imply that seeds contribute only negligibly to overall memory consumption.
That is because some data structures partition their input into many small parts called \emph{buckets} and store a seed for each of these buckets. The seeds may even make up the majority of the stored data.
This is the case for many constructions of perfect hash functions, but also occurs in a construction of compressed static functions \cite{belazzougui2013compressed}, which, in turn, enable space-efficient solutions to the approximate membership problem and the relative membership problem.
The way we phrased the SSEP implicitly assumes that buckets are handled independently. We discuss a slight generalisation in
\iffullversion
\cref{sec:SSSEP}.
\else
the full version \cite{lehmann2025combinedArxiv}.
\fi

\subparagraph{Minimal Perfect Hash Functions.}
\label{sec:mphf-results}
As an application of \consensus we discuss the construction of minimal perfect hash functions (MPHFs).
An MPHF on a set $X$ of keys maps these keys to a range $[|X|]$ without collisions. The function does not necessarily have to store $X$ explicitly. The space lower bound is $\OPTMPHF = |X|\log₂(e)-𝒪(\log |X|)$ bits assuming the universe of possible keys has size $Ω(|X|^2)$ \cite{belazzougui2009hash,mehlhorn1982program,mairson1983program,fredman1984bounds}.
While a succinct construction has long been known \cite{hagerup2001efficient}, its space overhead is fixed to $ε=\Theta(\frac{\log² \log n}{\log n})$ and large in practice.
Many recent practical approaches for constructing MPHFs involve a brute force search for seed values that map certain subsets of the keys in favourable ways \cite{pibiri2021pthash,esposito2020recsplit,hermann2024phobic,lehmann2023shockhash,lehmann2024shockhash2}. %
Those that do, all combine aspects of the MIN-strategy (store smallest working seeds) and the UNI-strategy (let a single seed do lots of work)
and achieve a space budget of $(1+ε)·\OPTMPHF$ only in time $|X|·\exp(Ω(1/ε))$ (see \cref{s:relatedMphf}). In this paper we break this barrier for the first time.
While we give the details in \cref{s:mphf}, the clean version of our result is the following.
\begin{restatable}[Bucketed \consensus-RecSplit]{theorem}{bucketedCRS}
    \label{thm:improved-recsplit-bucketed}
    For any $ε ∈ [n^{-3/7},1]$ there is an MPHF data structure with space requirement $(1+𝒪(ε))\OPTMPHF$, expected construction time $𝒪(n/ε)$ and expected query time $𝒪(\log 1/ε)$.
\end{restatable}

\subsection{Overview of the Paper}
The structure of the paper is as follows.
In \cref{sec:algorithm}, we explain \consensus in detail, as well as state the main theoretic results.
In \cref{sec:full-consensus-analysis,sec:simplified-consensus-analysis} we then prove the results.
We give an application of \consensus to minimal perfect hashing in \cref{s:mphf} and also briefly evaluate our practical implementation.
Finally, we conclude the paper in \cref{sec:conclusion}.

\section{Combined Search and Encoding for Successful Seeds}
\label{sec:algorithm}

In this section we present two variants of our \consensus algorithm. The first aims for conceptual clarity but produces seeds of length $Ω(n)$ bits that cannot be decoded efficiently. The second is more technical but solves this problem.

In the following we always assume that $n ∈ ℕ$, $p₁,…,pₙ ∈ (0,1]$ and $ε ∈ (0,1]$ are given and $\vB{i} = (\B{i}_j)_{j ∈ ℕ₀}$ is a Bernoulli process with parameter $p_i$ for all $i ∈ [n]$.

\subparagraph{A Branching Process.}
Let $k₁,…,kₙ ∈ ℕ$ be a sequence of numbers.\footnote{It may help to imagine that $k_i ≈ 2^ε/p_i$.} Consider the infinite tree visualised in \cref{fig:branching-process}. It has $n+2$ layers, indexed from $0$ to $n+1$. The root in layer~$0$ has an infinite number of children and nodes in layer $1 ≤ i ≤ n$ have $k_i$ children each, with $0$-based indexing in both cases. Edges between layer $i$ and $i+1$ are associated with the variables $\B{i}₀,\B{i}₁,\B{i}₂,…$ from top to bottom where $\B{i}_j = 0$ indicates a \emph{broken} edge depicted as a dotted line.
For convenience, we make another definition. For any $0 ≤ ℓ ≤ n$ and any sequence $(σ₀,σ₁,…,σ_ℓ) ∈ ℕ₀ × \{0,…,k₁-1\} × … × \{0,…,k_ℓ-1\}$ consider the node reached from the root by taking in layer $0 ≤ i ≤ ℓ$ the edge to the $σ_i$-th child. We define $σ₀∘…∘σ_ℓ$ to be the index of the reached node within layer $ℓ$.

\begin{figure}[bt]
    \centering
    \includegraphics[scale=0.9]{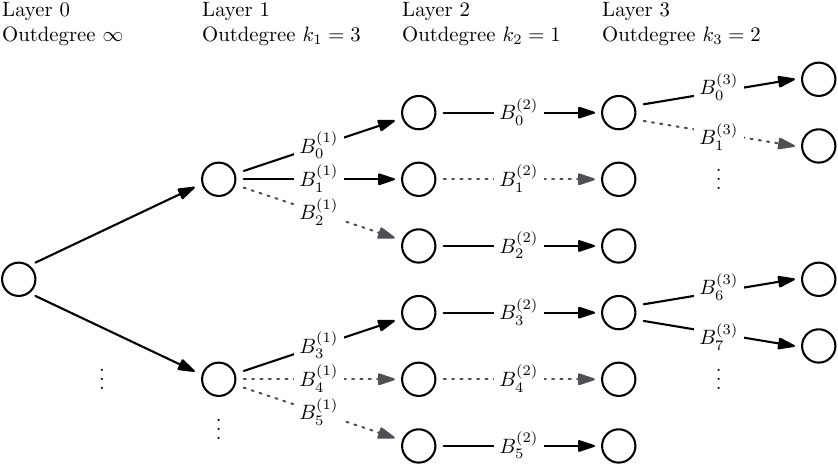}
    \caption{A visualisation of the tree underlying the search of the simplified \consensus algorithm.}
    \label{fig:branching-process}
\end{figure}

\subparagraph{The Simplified \consensus Algorithm.}
In \consensus, we search for a non-broken path from the root of the tree to a leaf node using a DFS, see \cref{algo:simplified-consensus}. Note that such a path exists with probability $1$ because the root has infinitely many children. The path is given by the sequence of choices $(σ₀,…,σₙ) ∈ ℕ₀ × \{0,…,k₁-1\} × … × \{0,…,kₙ-1\}$ made in each layer.
The returned sequence can be encoded as a bitstring $M$. Then $S_i \coloneq σ₀∘…∘σ_i$ is a successful seed for each $i ∈ [n]$ by construction. The merits of the approach are summarised in the following theorem.

\begin{figure}
    \setlength{\algomargin}{0em}
    \SetInd{5pt}{5pt}
    \begin{minipage}[t]{0.5\textwidth}
        \begin{algorithm}[H]
            \KwIn{$n,(p₁,…,pₙ),(k₁,…,kₙ)$ and\\
            \quad $\B{i}₀,\B{i}₁,… \sim \Ber(p_i)$ for all $i ∈ [n]$.}
            \KwOut{$(σ₀,…,σ_n)$ with $0 ≤ σ_i < k_i$
            \\\quad such that $\B{i}_{σ₀∘σ₁∘…∘σ_i} = 1$ for all $i ∈ [n]$.}

            \algo{Simplified-\consensus}{
                \For{$σ₀ ∈ \{0,1,2,…\}$}{
                    $(σ₁,i) ← (0,1)$\;
                    \While{$i > 0$}{
                        \uIf{$\B{i}_{σ₀∘…∘σ_i} = 1$}{
                            \If{$i = n$}{
                                \Return $M = (σ₀,…,σₙ)$\;
                            }
                            $i ← i+1$\;
                            $σ_i ← 0$\;
                        }\Else(\tcp*[h]{backtrack and/or next seed}){
                            \While{$i > 0$ \And $σ_i = k_i-1$}{
                                $i ← i-1$\;
                            }
                            \If{$i > 0$}{
                                $σ_i ← σ_i+1$
                            }
                        }
                    }
                }
            }
            \caption{Simplified \consensus.}
            \label{algo:simplified-consensus}
        \end{algorithm}
    \end{minipage}
    \begin{minipage}[t]{0.5\textwidth}
        \begin{algorithm}[H]
            \KwIn{$n,w,(p₁,…,pₙ),(ℓ₁,…,ℓₙ)$ and\\\quad $\B{i}₀,…,\B{i}_{2^w-1} \sim \Ber(p_i)$ for all $i ∈ [n]$.}
            \KwOut{either $⊥$ or $M$ as in \cref{fig:full-consensus-notation}\\\quad such that $\B{i}_{S_i} = 1$ for all $i ∈ [n]$.}
            
            \algo{Full-\consensus}{
                \For(\tcp*[h]{bounded}){$σ₀ ∈ \{0,1,…,2^{w}-1\}$}{
                    $(σ₁,i) ← (0,1)$\;
                    \While{$i > 0$}{
                        \uIf(\tcp*[h]{$S_i$ as in \cref{fig:full-consensus-notation}}){$\B{i}_{S_i} = 1$}{
                            \If{$i = n$}{
                                \Return $M = (σ₀,…,σₙ)$\;
                            }
                            $i ← i+1$\;
                            $σ_i ← 0$\;
                        }\Else(\tcp*[h]{backtrack and/or next seed}){
                            \While{$i > 0$ \And $σ_i = 2^{ℓ_i}-1$}{
                                $i ← i-1$\;
                            }
                            \If{$i > 0$}{
                                $σ_i ← σ_i+1$
                            }
                        }
                    }
                }
                \Return $⊥$ \tcp{can fail}
            }
            \caption{Full \consensus.}
            \label{algo:full-consensus}
        \end{algorithm}
    \end{minipage}
\end{figure}

\begin{theorem}[Performance of Simplified \consensus]
    \label{thm:simplified-consensus}
    Given $n ∈ ℕ$, $p₁,…,pₙ ∈ (0,1]$ and $ε ∈ (0,1]$ let $k₁,…,kₙ ∈ ℕ$ be a sequence of numbers satisfying\footnote{Ideally we would want $k_j \coloneq \frac{2^ε}{p_j}$ but the integer constraint on the $k_j$ forces us to be more lenient.}
    $\prod_{j = 1}^{i} p_j k_j 2^{-ε} ∈ [1,2]$ for all $i ∈ [n]$.
    Then \cref{algo:simplified-consensus} solves the SSEP with $𝔼[T_i] = 𝒪(T_i^\OPT / ε)$ and $𝔼[|M|] ≤ M^\OPT+εn+\log₂ 1/ε + 𝒪(1)$.
\end{theorem}

\subparagraph{Full \consensus Algorithm.}
The main issues with the simplified \consensus algorithm are that $(σ₀,…,σₙ)$ is awkward to encode in binary and that using $σ₀∘…∘σ_i$ directly as the seed $S_i$ makes for seeds of size $Ω(n)$ bits, which are inefficient to compute and handle. We solve these issues with three interventions.
\begin{enumerate}
    • We introduce a word size $w ∈ ℕ$ and encode $σ₀$ using $w$ bits (hoping that it does not overflow, see below).
    • We demand that $k₁,…,kₙ ∈ ℕ$ are powers of two, i.e.\ $k_i = 2^{ℓ_i}$ for $ℓ_i ∈ ℕ₀$. With our definition of “$∘$”, the binary representation of the number $σ₀∘…∘σ_i$ now simply arises by concatenating the binary representations of $σ₀,…,σ_i$ (where $σ_j$ is encoded using $ℓ_j$ bits and $σ₀$ is encoded using $w$ bits).
    • Instead of defining $S_i$ to be $σ₀∘…∘σ_i$, we define $S_i$ to be the $w$-bit suffix of $σ₀∘…∘σ_i$.
    We will show that this truncation is unlikely to cause collisions, i.e.\ any two sequences $(σ₀,…,σ_i)$ and $(σ'₀,…,σ'_i)$ we encounter have distinct suffixes $S_i ≠ S_i'$.
    Intuitively, this guarantees that the search algorithm will always see “fresh” random variables and does not “notice” the change.    
\end{enumerate}

\begin{figure}[tb]
    \centering
    \includegraphics[scale=0.9]{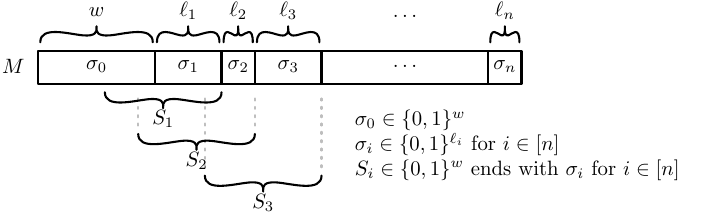}
    \caption{The bitstring $M$ returned by the full \consensus algorithm and the notation used to refer to relevant substrings.}
    \label{fig:full-consensus-notation}
\end{figure}

The full algorithm is given as \cref{algo:full-consensus}, using notation defined in \cref{fig:full-consensus-notation}.
Our formal result is as follows.

\begin{theorem}[Performance of Full \consensus]
    \label{thm:full-consensus}
    Let $c > 0$ be an arbitrary constant. Given $n ∈ ℕ$, $p₁,…,pₙ ∈ (n^{-c},1]$ and $ε ∈ (n^{-c},1]$, let $ℓ₁,…,ℓₙ ∈ ℕ₀$ be the unique sequence of numbers such that
    \begin{gather*}
        \sum_{j = 1}^i ℓ_j \coloneq \Big{⌈}εi + \sum_{j = 1}^i \log₂(1/p_i)\Big{⌉} \text{ for $i ∈ [n]$.}
    \end{gather*}
    Assume we run \cref{algo:full-consensus} with these parameters and $w ≥ (2+7c)\log₂ n$.\footnote{We did not attempt to improve the constant “$2+7c$” and expect that $w = 64$ is enough in practice.} Then there is an event $E$ with $\Pr[E] = 1-𝒪(n^{-c})$ that implies that the run solves the SSEP, producing a bitstring $M$ with $|M| = ⌈M^\OPT+εn+w⌉$ containing $S_i ∈ \{0,1\}^w$ at offset $\sum_{j = 1}^i ℓ_j$ for all $i ∈ [n]$. Moreover, $𝔼[T_i \mid E] = 𝒪(T_i^\OPT/ε)$ for all $i ∈ [n]$.
\end{theorem}
In particular the result guarantees that $S_i$ can be inferred efficiently given $M$ and $i$ provided that $\sum_{j = 1}^i \log₂(1/p_i)$ can be computed efficiently given $i$. In practice it is fine to use approximations of $p_i$.\footnote{Our implementation uses fixed point arithmetic (integers counting “microbits”) to avoid issues related to floating point arithmetic not being associative and distributive.}
Note that the expectation of $T_i$ is only bounded conditioned on a good event $E$, which is defined in \cref{sec:full-consensus-analysis}. To obtain an unconditionally bounded expectation, \cref{algo:full-consensus} should be modified to give up early in failure cases.\footnote{To do so we need not decide whether $E$ occurs. We can simply
count the number of steps and abort when a limit is reached.
}

\subparagraph{Special Cases and Variants.}
An important special case is $p₁ = … = pₙ = p$. With $b = \log₂ 1/p$ the space lower bound is then simply $M^\OPT = bn$ bits. If we apply \cref{thm:full-consensus} with $ε ∈ (0,1)$ then the lengths $ℓ₁,…,ℓₙ$ of the seed fragments $σ₁,…,σₙ$ are all either $⌈b+ε⌉$ or $⌊b+ε⌋$ with an average approaching $b+ε$. Assuming $b ∉ ℕ$ it may be tempting to use $ε = ⌈b⌉-b$ such that all seed fragments have the same length $b+ε ∈ ℕ$. Similarly, we could pick $ε$ such that $b+ε$ is a “nice” rational number, e.g.\ $b+ε = b'+\frac 12$ for $b' ∈ ℕ$ such that $ℓ₁,…,ℓₙ$ alternate between $b'+1$ and $b'$.

In some cases it may be reasonable to choose $ε > 1$. Even though \consensus ceases to be more space efficient than MIN, it still has the advantage that seeds are bitstrings of fixed lengths in predetermined locations. We decided against formally covering this case in \cref{thm:full-consensus}, but argue for the following modification in
\iffullversion
\cref{sec:large-eps}.
\else
the full version \cite{lehmann2025combinedArxiv} of this paper.
\fi

\def\strike#1{%
    \begin{tikzpicture}[baseline=(t.base)]
        \node[inner sep=0] (t) {#1};
        \draw (t.south west) -- (t.north east);
    \end{tikzpicture}%
}
\begin{remark}
    \label{rem:large-eps}
    The guarantees of \cref{thm:full-consensus} also hold for $ε ∈ [1,\log n]$, except that
    \[ 𝔼[T_i \mid E] = T_i^\OPT(1+\exp(-Ω(2^ε)) + n^{-c}) \qquad \text{\color{gray} (instead of \strike{$𝔼[T_i \mid E] = 𝒪(T_i^\OPT/ε)$}}).\]
\end{remark}
It may also be worthwhile to consider variants of \consensus that use different values of the overhead parameter $ε$ in different parts of $M$ to minimise some overall cost.
Such an idea makes sense in general if we augment the SSEP to include specific costs $c_i$ associated with testing the Bernoulli random variables belonging to the $i$th family.

\section{Analysis of Simplified \consensus}
\label{sec:simplified-consensus-analysis}

We begin by reducing the claim of \cref{thm:simplified-consensus} to a purely mathematical statement given in \cref{lem:simplified-consensus}, which we prove in \cref{sec:simplified-consensus-lemma}.

\subsection{Reduction of Theorem \ref{thm:simplified-consensus} to Lemma \ref{lem:simplified-consensus}}

Recall the tree defined in \cref{sec:algorithm} and illustrated in \cref{fig:branching-process}. For any $i ∈ [n+1]$ consider a node $v$ in layer $i$. Let $q_i$ be the probability that there is an unbroken path from $v$ to a node in layer $n+1$. (This probability is indeed the same for all nodes in layer $i$.)
In the last layer, we have $q_{n+1} = 1$ by definition. For all $1 ≤ i ≤ n$ we have $q_i = 1-(1-p_iq_{i+1})^{k_i}$ since $p_i q_{i+1}$ is the probability that an unbroken path via a specific child of $v$ exists, and since there is an independent chance for each of the $k_i$ children.
The technical challenge of this section is to show that the $q_i$ are not too small.

\begin{lemma}
    \label{lem:simplified-consensus}
    Under the conditions of \cref{thm:simplified-consensus} we have $q_i = Ω(ε)$ for all $i ∈ [n]$.
\end{lemma}
This lemma implies \cref{thm:simplified-consensus} as follows.

\begin{proof}[Proof of \cref{thm:simplified-consensus}]
    Let $T_i$ for $i ∈ [n]$ be the number of Bernoulli random variables from the family $(B_j^{(i)})_{j ∈ ℕ₀}$ that are inspected. This corresponds to the number of edges between layers $i$ and $i+1$ that are inspected by the DFS. With probability $p_i$ such an edge is unbroken and, if it is unbroken, then with probability $q_{i+1}$ the DFS will never backtrack out of the corresponding subtree and hence never inspect another edge from the same layer. This implies $T_i \sim \Geom(p_i q_{i+1})$ and thus $𝔼[T_i] = \frac{1}{p_i q_{i+1}} = 𝒪(T_i^\OPT/ε)$ where the last step uses $T_i^\OPT = 1/p_i$ and \cref{lem:simplified-consensus}.
    
    Concerning the space to encode $M = (σ₀,…,σₙ)$, there are two parts to consider. The “fixed width” numbers $(σ₁,…,σₙ) ∈ \{0,…,k₁-1\} × … × \{0,…,kₙ-1\}$ can be encoded using $⌈\log₂ \prod_{i ∈ [n]} k_i⌉$ bits.
    By the assumption on $(k_i)_{i ∈ [n]}$ we have $\prod_{i ∈ [n]} k_i ≤ 2·\prod_{i ∈ [n]} 2^{ε}/p_i$ and can bound the space by
    \[
          1+ \log₂ \prod_{i ∈ [n]} k_i
        ≤ 1+ \log₂ \Big(2 \prod_{i ∈ [n]} \frac{2^ε}{p_i}\Big)
        ≤ 2 + εn + \sum_{i ∈ [n]} \log₂ \frac{1}{p_i} = M^\OPT + εn + 𝒪(1).
    \]
    The “variable width” number $1+σ₀ ∈ ℕ$ has distribution $1+σ₀ \sim \Geom(q₁)$ because $σ₀$ is the $0$-based index of the first node in layer $1$ with an unbroken path to layer $n+1$. Again by \cref{lem:simplified-consensus} we have $𝔼[1+σ₀] = \frac{1}{q₁} = 𝒪(1/ε)$. Using Jensen's inequality \cite{jensen1906fonctions,mitzenmacher2017probability} and the fact that $\log_2$ is concave, the expected number of bits to encode $σ₀$ is hence bounded by
    \[ 𝔼[⌈\log₂(1+σ₀)⌉] ≤ 1 + \log₂ 𝔼[1+σ₀] = 1 + \log₂(𝒪(1/ε)) = \log₂ 1/ε + 𝒪(1). \]
    Summing the contributions of both parts yields a bound on $𝔼[|M|]$ as claimed.
\end{proof}

\subsection{Proof of Lemma \ref{lem:simplified-consensus}}

To understand the proof of \cref{lem:simplified-consensus} it may help to consider a simplified case first.
Assume $p₁ = … = pₙ = p$ and $k₁ = … = kₙ = k$ with $kp = 2^ε$, i.e.\ $k$ is slightly larger than $\frac{1}{p}$. We have $q_i = 1-(1-pq_{i+1})^k$, which motivates defining $f(x) = 1-(1-px)^k$ such that $q_i$ arises by iteratively applying $f$ to an initial value of $q_{n+1} = 1$.
From the illustration in \cref{fig:iterating-a-function} we see that $q_i$ can never drop below the largest fixed point $x^*$ of $f$, and it is not hard to show that $x^* = Ω(ε)$.

\begin{figure}[tb]
    \centering
    \includegraphics[scale=0.9]{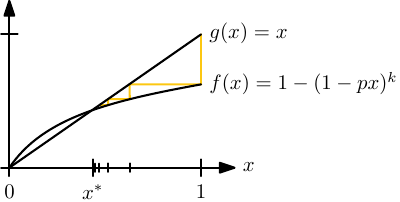}
    \caption{For $p₁ = … = pₙ = \frac{1}{4}$ and $k₁ = … = kₙ = 5$ (thus $ε = \log₂(5/4) ≈ 0.32$) the values $q_{n+1},q_{n},q_{n-1},…$ (marked as ticks on the $x$-axis) arise by iteratively applying $f(x) = 1-(1-px)^k$ to the starting value of $1$. Each $q_i$ is at least the fixed point $x^* ≈ 0.45$ of $f$.}
    \label{fig:iterating-a-function}
\end{figure}

Our full proof has to deal with a setting that is much less clean. The values $p_i$ and $k_i$ depend on $i$, yielding a distinct function $f_i$ for each $i ∈ [n]$ and $p_i k_i ≈ 2^ε$ is only true in the sense of a geometric average over $Ω(1/ε)$ consecutive indices. We hence analyse compositions $f_{a}∘…∘f_{b-1}$ for sufficiently large $b-a = Θ(1/ε)$.
\label{sec:simplified-consensus-lemma}%
We begin with some definitions.

\begin{itemize}
    • Let $f_i(x) = 1-(1-p_i x)^{k_i}$. Note that this gives $q_i = (f_i ∘ … ∘ fₙ)(1)$ for $i ∈ [n+1]$.
    • For any set $I = \{a,…,b-1\}$ let
    $P_I = \prod_{i ∈ I} p_i$, $K_I = \prod_{i ∈ I} k_i$, and $F_I = f_a ∘ … ∘ f_{b-1}$.
    Note that the assumption of \cref{thm:simplified-consensus} gives
    \[
        P_I K_I ·2^{-ε|I|} =
        \underbrace{P_{\{1,…,b-1\}} K_{\{1,…,b-1\}} 2^{-ε(b-1)}}_{∈ [1,2]}
        \Big( \underbrace{P_{\{1,…,a-1\}} K_{\{1,…,a-1\}} 2^{-ε(a-1)}}_{∈ [1,2]} \Big)^{-1} ∈ [\tfrac 12,2].
    \]
    • Let $ε₀ \coloneq ⌈2ε^{-1}⌉^{-1}$. Its role is to be a number satisfying
    $ε₀ = Θ(ε)$, $ε₀^{-1} ∈ ℕ$, and $(2^{ε})^{ε₀^{-1}} ∈ [4,8]$.
    In particular, if $I = \{a,…,b-1\}$ with $b-a = ε₀^{-1}$ then
    \[
        P_I K_I = \underbrace{P_I K_I 2^{-ε|I|}}_{∈ [\tfrac 12,2] } · \underbrace{2^{ε|I|}}_{∈ [4,8]} ∈ [2,16].
    \]
    • Lastly, let $C = 8$ and $ε₁ = ε₀/C$.
\end{itemize}

These definitions are designed to make the proof of the following lemma as simple as possible. No attempt was made to optimise constant factors.

\begin{lemma}
    \label{lem:iterative-applications-of-functions}
    Let $I = \{a,…,b-1\} ⊆ [n]$. The following holds.
    \begin{enumerate}[(i)]
        • For all $i ∈ [n]$ and $x ∈ [0,1]$: $f_i(x) ≥ p_i k_i x - (p_i k_i x)²/2$.
        • If $b-a ≤ ε₀^{-1}$ then $F_I(x) ≥ P_I K_I x (1-C|I|x)$ for $x ∈ [0,ε₁]$.
        • If $b-a = ε₀^{-1}$ then $F_I(ε₁/2) ≥ ε₁/2$.
    \end{enumerate}
\end{lemma}
\begin{proof}\ 
    \begin{enumerate}[(i)]
        • $e^{-x}$ satisfies the following bounds: $1-x ≤ e^{-x} ≤ 1 - x + x²/2$ for $x > 0$.
        We get the left claim by integrating $e^{-x} ≤ 1$ for $t ∈ [0,x]$ to get $1-e^{-x} ≤ x$. Integrating again for $t ∈ [0,x]$ gives the right claim.
        Applying to $f_i(x)$ we get
        \[ f_i(x) = 1-(1-p_i x)^{k_i} ≥ 1 - e^{-p_i k_i x} ≥ 1 - (1-p_i k_i x + (p_i k_i x)²/2) = p_i k_i x - (p_i k_i x)²/2.\]
        • We use induction on $b-a$. For $b - a= 0$ we have $I = ∅$, $F_I = \mathrm{id}$ as well as $P_I = K_I = 1$ and the claim holds. Let now $I = \{a, …, b-1\}$ with $b-a ≤ ε₀^{-1}$and $I' = \{a+1,…,b-1\}$. We assume the claim holds for $I'$. Let $x ∈ [0,ε₁]$.
        \begin{align*}
            F_I(x)
            &= (f_a ∘ F_{I'})(x)
            = \overset{\mathclap{\underset{↓}{\smash{\text{monotonic}}}}}{f_a}(F_{I'}(x))
            \textrel{induction}{≥} f_a(P_{I'} K_{I'} x (1-C|I'|x))
            \\
            &\textrel{(i)}{≥} p_a k_a P_{I'} K_{I'} x (1-C|I'|x)
            - (p_a k_a P_{I'} K_{I'} x (1-C|I'|x))²/2\\
            &= P_{I} K_{I} x (1-C|I'|x)
            - (P_{I} K_{I} x (1-\underbrace{C|I'|x}_{\mathclap{≤ Cε₀^{-1}ε₁ = 1 < 2
            \rlap{\ and hence $(1-C|I'|x)² ≤ 1$}
            }}))²/2\\
            &≥ P_{I} K_{I} x (1-C|I'|x)
            - (P_{I} K_{I} x)²/2
            = P_{I} K_{I} x (1-x(C|I'|+\underbrace{P_I K_I}_{\mathclap{≤ 16 = 2C}} / 2))\\[-7pt]
            &≥ P_{I} K_{I} x (1-x(C|I'|+C))
            = P_{I} K_{I} x (1-C|I|x).
        \end{align*}
        • By applying (ii) and using $Cε₀^{-1}ε₁ = 1$ we get
        \[ F_I(ε₁/2) \textrel{\smash{(ii)}}{≥} \underbrace{P_I K_I}_{≥ 2} \tfrac{ε₁}{2} \underbrace{(1-Cε₀^{-1} \tfrac{ε₁}{2})}_{= \frac 12} ≥ ε₁/2.\qedhere\]
    \end{enumerate}
\end{proof}
The proof of \cref{lem:simplified-consensus} is now immediate:
\begin{proof}[Proof of \cref{lem:simplified-consensus}]
    Let $i ∈ [n]$. Partition $I = \{i,…,n\}$ into contiguous intervals $I₁,…,I_k$ such that $F_I = F_{I₁}∘ … ∘ F_{I_k}$ with $|I₁| ≤ ε₀^{-1}$ and $|I_j| = ε₀^{-1}$ for $j ∈ \{2,…,k\}$.
    Then by repeated applications of \cref{lem:iterative-applications-of-functions} (iii), one application of \cref{lem:iterative-applications-of-functions} (ii) and the monotonicity of the relevant functions we get
    \begin{align*}
        q_i &= F_I(1) ≥ F_I(ε₁/2) = F_{I₁}(F_{I₂}(…(F_{I_k}(ε₁/2))…)) \textrel{(iii)}{≥} F_{I₁}(ε₁/2)\\
        &
        \textrel{(ii)}{≥} P_{I₁} K_{I₁} \frac{ε₁}{2} (1-C|I₁|\frac{ε₁}{2})
        = \underbrace{P_{I₁} K_{I₁} 2^{-ε|I₁|}}_{≥ \frac 12} \underbrace{2^{ε|I₁|}}_{≥1} \frac{ε₁}{2} (1-\underbrace{C|I₁|\frac{ε₁}{2}}_{≤ \frac 12})
        ≥ \frac{ε₁}{8}.\qedhere
    \end{align*}
\end{proof}

\section{Analysis of Full \consensus}
\label{sec:full-consensus-analysis}

\def\prefix{\mathrm{prefix}}
\def\suffix{\mathrm{suffix}}

Full \consensus performs a DFS in a tree as depicted in \cref{fig:branching-process} in much the same way as simplified \consensus, with two differences. First, the root node only has $2^w$ children (not infinitely many) and the probability that an unbroken path exists is therefore strictly less than $1$. Secondly, when examining the outgoing edge of index $σ_i$ of some node in layer $i$, then its status is no longer linked to $B_{Σ}^i$ for $Σ = σ₀∘…∘σ_i$. Instead $Σ$ is divided as $Σ = \prefix(Σ) ∘ \suffix(Σ)$ where $\suffix(Σ) ∈ \{0,1\}^w$ and $B_{\suffix(Σ)}^i$ is used.
To prepare for the proof of \cref{thm:full-consensus}, we present a coupling to, and a claim about simplified \consensus.

\subparagraph{A Coupling between Simplified and Full \consensus.}
Consider the following natural coupling. Let $B_j^{(i)}$ for $i ∈ [n]$ and $j ∈ ℕ₀$ be the random variables underlying simplified \consensus and $\tilde{B}_{S}^{(i)}$ for $i ∈ [n]$ and $S ∈ \{0,1\}^w$ the random variables underlying full \consensus. The coupling should ensure that $B^{(i)}_{Σ} = \tilde{B}^{(i)}_{\suffix(Σ)}$ for all pairs $(i,Σ)$ such that simplified \consensus inspects $B^{(i)}_{Σ}$ and does not at an earlier time inspect $B^{(i)}_{Σ'}$ with $\suffix(Σ') = \suffix(Σ)$. In other words, for any pair $(i,S)$, if simplified consensus inspects at least one variable $B^{(i)}_{Σ}$ with $\suffix(Σ) = S$ then $\tilde{B}^{(i)}_S$ equals the earliest such variable that is inspected.

This implies that the runs of simplified and full \consensus behave identically as long as, firstly, $σ₀$ never reaches $2^w$ and, secondly, simplified \consensus never \emph{repeats a suffix}, by which we mean inspecting $B^{(i)}_Σ$ for some pair $(i,Σ)$ such that $B^{(i)}_{Σ'}$ with $\suffix(Σ) = \suffix(Σ')$ was previously inspected. In the following, by a \emph{step} of \consensus we mean one iteration of its main while-loop, which involves inspecting exactly one random variable.

\begin{claim}
    \label{claim:repeating-suffix-prob}
    Simplified \consensus repeats a suffix in step $t$ with probability at most $\frac{t}{2^w}$.
\end{claim}
\begin{proof}
    Assume simplified \consensus is about to inspect $B^{(i)}_{Σ^*}$. Call $S ∈ \{0,1\}^w$ a \emph{blocked suffix} if there exists $Σ ∈ \{0,1\}^*$ such that $S = \suffix(Σ)$ and $B^{(i)}_{Σ}$ was previously inspected. This necessitates $\prefix(Σ) < \prefix(Σ^*)$. Clearly, step $t$ repeats a suffix if $\suffix(Σ^*)$ is blocked.
    A crucial observation is that the DFS has exhaustively explored all branches of the tree associated with any prefix less than $Σ^*$ and the order in which this exploration has been carried out is no longer relevant. The symmetry between all nodes of the same layer therefore implies that every suffix is blocked with the same probability. Since at most $t$ suffixes can be blocked at step $t$, the probability that $Σ^*$ is blocked (given our knowledge of $(i,Σ^*)$) is at most $\frac{t}{2^w}$.
\end{proof}

\begin{proof}[Proof of \cref{thm:full-consensus}.]
    \def\err{\mathrm{err}}
    \def\limit{\mathrm{limit}}
    First note that the choice of $(ℓ₁,…,ℓₙ)$ made in \cref{thm:full-consensus} amounts to a valid choice of $(k₁ = 2^{ℓ₁},…,kₙ = 2^{ℓₙ})$ for applying \cref{thm:simplified-consensus}. Indeed, for each $i ∈ [n]$ there exists some $\err_i ∈ [1,2)$ related to rounding up such that
    \begin{align*}
        \prod_{j = 1}^i p_j k_j 2^{-ε}
        &= \Big( \prod_{j = 1}^i p_j \Big) 2^{\sum_{j = 1}^{i} ℓ_j} 2^{-εi}
        = \Big( \prod_{j = 1}^i p_j \Big) 2^{⌈εi + \sum_{j = 1}^i \log₂(1/p_i)⌉} 2^{-εi}\\
        &= \Big( \prod_{j = 1}^i p_j \Big) 2^{εi + \sum_{j = 1}^i \log₂(1/p_i)} ·\err_i 2^{-εi}
        = \err_i ∈ [1,2]
    \end{align*}
    as required. The number $T = \sum_{i = 1}^{n} T_i$ of steps taken by simplified \consensus satisfies $𝔼[T] = 𝒪(\sum_{i = 1}^n \frac{1}{p_i ε}) = 𝒪(n^{1+2c})$. By Markov's inequality $T$ exceeds its expectation by a factor of $n^c$ with probability at most $n^{-c}$ so with probability at least $1-𝒪(n^{-c})$ simplified \consensus terminates within $T_{\limit} = n^{1+3c}$ steps.
    
    By \cref{claim:repeating-suffix-prob} and a union bound the probability that simplified \consensus repeats a suffix within its first $T_{\limit}$ steps is at most $\sum_{t = 1}^{T_{\limit}} t/2^w = 𝒪(T_{\limit}²/2^w) = 𝒪(n^{2+6c}/2^w)$. By choosing $w ≥ (2+7c)\log₂ n$ this probability is $𝒪(n^{-c})$. This choice also guarantees that $σ₀$ never exceeds $2^w$ during the first $T_{\limit}$ steps. Let now $E = \{T ≤ T_{\limit}\} ∩ \{\text{no repeated suffix in the first $T_{\limit}$ steps}\}$. By a union bound $\Pr[E] = 1-𝒪(n^{-c})$ and by our coupling $E$ implies that simplified and full \consensus explore the same sequence of nodes in their DFS and terminate with the same sequence $σ₀,…,σₙ$ with $σ₀ < 2^w$. In particular full \consensus succeeds in that case as claimed.
    
    The updated bound on $|M|$ simply reflects that $σ₀$, which in simplified \consensus had expected length $\log₂ 1/ε + 𝒪(1)$ now has fixed length $w = 𝒪(\log n)$. Lastly, if $T_i$ and $\tilde{T}_i$ denote the number of Bernoulli trials from the $i$-th family inspected by simplified and full \consensus, respectively, then conditioned on $E$ we have $T_i = \tilde{T}_i$ so
    \[ 𝔼[\tilde{T}_i \mid E] = 𝔼[T_i | E] ≤ 𝔼[T_i]/\Pr[E] = 𝒪(𝔼[T_i]) = 𝒪(T_i^\OPT / ε).\qedhere\]
\end{proof}

\section{Application to Minimal Perfect Hashing}\label{s:mphf}
In this section, we apply the \consensus idea to the area of minimal perfect hash functions (MPHFs).
Recall from \cref{sec:mphf-results} that an MPHF for a given set $X$ of $n$ keys maps each key in $X$ to a unique integer from $[n]$.

\subsection{Existing Strategies for Constructing MPHFs}\label{s:relatedMphf}

We begin with a brief overview of established approaches with a focus on those we need. Detailed explanations are found in \cite{lehmann2024fast,lehmann2025modern}.

\subparagraph{Partitioning.}
Most approaches randomly partition the input into small \emph{buckets} using a hash function and then construct an MPHF on each bucket separately.
The natural way to obtain an MPHF on the entire input involves storing the bucket-MPHFs and prefix sums for the bucket sizes.
Storing the prefix sums can be avoided when the partitioning uses a $k$-perfect hash function, as is done in ShockHash-Flat \cite{lehmann2024shockhash2}, and as we do in \cref{sec:bucketed_consensus_recsplit}.

Hagerup and Tholey \cite{hagerup2001efficient} demonstrate that partitioning into tiny buckets of size $𝒪(\frac{\log n}{\log \log n})$ when suitably combined with exhaustive tabulation of MPHFs on tiny inputs yields an approach with construction time $𝒪(n)$, query time $𝒪(1)$ and space overhead $ε = \Theta(\frac{\log² \log n}{\log n})$. Unfortunately the analysis requires astronomical values of $n ≥ \max(\exp(ω(1/ε)),2^{150})$ \cite{botelho2013practical} and is restricted to these rather large $ε$.

Below we explain \emph{monolithic} variants of MPHFs, i.e.\ constructions not using partitioning, which are typically combined with partitioning to obtain faster \emph{bucketed} variants.

\subparagraph{Recursive Splitting.}
A recursive splitting strategy \cite{esposito2020recsplit} involves a predefined \emph{splitting tree}, which is a rooted tree with $n$ leaves.\footnote{In the original paper \cite{esposito2020recsplit}, the lowest level is suppressed, with correspondingly changed terminology.}
See \cref{fig:splitting-tree} (left) for an example with $n = 11$.
An MPHF for an input set of size $n$ is given by a family of functions containing for each internal node $v$ of the tree a \emph{splitting hash function} $f_v$ that maps keys to the children of $v$.
Together, the family $(f_v)_v$ must describe a recursive partitioning of the input set into parts of size~$1$.
The MPHF is queried for a key $x$ by starting at the root.
The splitting hash functions are evaluated on $x$ and the resulting path is followed until a leaf is reached, the index of which is returned.
Each $f_v$ is identified by a seed that is found using trial and error.
In the original paper \cite{esposito2020recsplit} seeds are stored using Golomb-Rice coding \cite{golomb1966run,rice1979some}.
To keep the space overhead low, the last layer of the splitting tree uses $\ell$-way splits for relatively large $\ell$ (5--16 in practice), which makes splitting hash functions costly to find.

\subparagraph{Multiple-Choice Hashing.}
We randomly assign a set $\{h₁(x),…,h_{k}(x)\} ⊆ [n(1+ε)]$ of candidates to each key hoping that a choice $σ(x) ∈ \{1,…,k\}$ can be made for each key such that $x ↦ h_{σ(x)}(x)$ is injective. A perfect hash function is then given by a retrieval data structure that stores $σ$ \cite{BPZ:Practical:2013,Vigna:Fast-Scalable-Construction-of-Functions:2016,LSW:SicHash:2023}.
A recent variant ShockHash \cite{lehmann2023shockhash,lehmann2024shockhash2} achieves record space efficiency by using the aggressive configuration $ε = 0$ and $k = 2$ (combined with partitioning and recursive splitting). Many retries are needed until $σ$ exists.

\subparagraph{Bucket Placement.}
A general strategy with many variants assigns a random bucket $f(x) ∈ \{1,…,b\}$ to each key and considers for each bucket $i ∈ [b]$ several (hash) functions $g_{i,1},g_{i,2},…$ with range $[n]$. A choice $σ(i)$ is stored for each bucket such that $x ↦ g_{f(x),σ(f(x))}(x)$ is a bijection. Typically, the values $σ(i)$ are chosen greedily in decreasing order of bucket sizes \cite{pibiri2021pthash,hermann2024phobic,belazzougui2009hash}, but backtracking has also been considered \cite{yang1985backtracking,fox1992practical} (in a different sense than \consensus).

\begin{figure}[t]
    \includegraphics{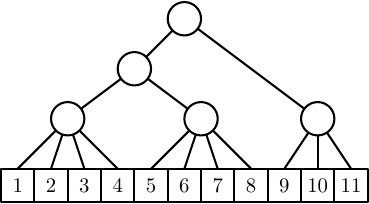}
    \includegraphics{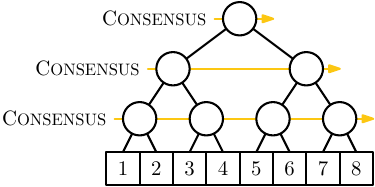}
    \caption{A general splitting tree and a balanced binary splitting tree as used in \cref{sec:monolithic-recsplit}.}
    \label{fig:splitting-tree}
\end{figure}

\subsection{Monolithic \consensus-RecSplit}
\label{sec:monolithic-recsplit}
We now describe a variant of recursive splitting in an idealised setting with $n = 2^d$ that uses \consensus to be particularly space efficient.
The idea is simple. The data structure consists of $n-1$ two-way splitting hash functions, i.e.\ hash functions with range $\{0,1\}$, arranged in a balanced binary tree with $n$ leaves, see \cref{fig:splitting-tree} (right).

This means that during construction for a key set $S$ we have to select the splitting hash functions from top to bottom such that each function splits the part of $S$ arriving at its node perfectly in half.
The data structure is entirely given by a sequence of $n-1$ seeds that identify the splitting hash functions.
If we search and encode these seeds using the \consensus algorithm, we obtain the following result.

\begin{theorem}[Monolithic \consensus-RecSplit]
    \label{thm:improved-recsplit-monolithic}
    Let $n = 2^d$ for $d ∈ ℕ$ and $ε ∈ [\frac{1}{n},1]$. Monolithic-RecSplit is an MPHF data structure with space requirement $\OPTMPHF+𝒪(εn+\log² n)$, expected construction time $𝒪(\max(n^{3/2},\frac{n}{ε}))$ and query time $𝒪(\log n)$.
\end{theorem}

\begin{proof}
    We use a separate \consensus data structure per level of the splitting tree and construct them from top to bottom.
    Consider level $ℓ ∈ \{0,…,d-1\}$ of the tree, which contains $2^ℓ$ nodes responsible for $n_ℓ = 2^{d-ℓ}$ keys each. The probability that a given seed for a splitting hash function amounts to a perfect split for $n_ℓ$ keys is $p(n_ℓ) = \binom{n_ℓ}{n_ℓ/2}·2^{-n_ℓ}$ by a simple counting argument. A standard approximation gives $p(n_ℓ) = Θ(1/\sqrt{n_ℓ})$.

    By applying \cref{thm:full-consensus} with $p₁ = … = p_{2^ℓ} = p(n_ℓ)$, $ε_ℓ = \min(1,ε·n_ℓ^{3/4})$ and $w = 𝒪(\log n)$ we obtain a data structure encoding $2^ℓ$ successful seeds. The expected number of seeds inspected is $𝒪(\frac{1}{p(n_ℓ) ε_ℓ})$ for each node, which corresponds to expected work of $𝒪(\frac{n_ℓ}{p(n_ℓ)ε_ℓ}) = 𝒪(n_ℓ^{3/2}/ε_ℓ)$ per node when seeds are tested by hashing all $n_ℓ$ relevant keys.\footnote{Dividing the work per node by $ε_ℓ$ gives $-n_ℓ^{3/2}/ε_ℓ²$. Our choice of $ε_ℓ \sim n_ℓ^{3/4}$ (as long as $\min(1,·)$ does not kick in) is desired to ensure that using extra bits amounts to roughly the same reduction of work in each level, such that we have a consistent trade-off between time and space.}
    For the total expected work $w_ℓ$ in level $ℓ$ there are, up to constant factors, the following two options.
    \begin{align*}
        \text{if $ε_ℓ = ε·n_ℓ^{3/4}$ then\ } & w_ℓ = \frac{n_ℓ^{3/2}}{ε n_ℓ^{3/4}}·2^ℓ = \frac{n_ℓ^{3/4}2^ℓ}{ε} = \frac{2^{(d-ℓ)·3/4 + ℓ}}{ε} = \frac{n^{3/4}·2^{ℓ/4}}{ε}.\\
        \text{if $ε_ℓ = 1$ then\ } & w_ℓ = n_ℓ^{3/2}·2^ℓ = 2^{(d-ℓ)·3/2+ℓ} = n^{3/2}·2^{-ℓ/2}.
    \end{align*}
    The formula for the first case is maximised for $ℓ = d-1$ giving $w_{d-1} = Θ(n/ε)$. The formula for the second case is maximised for $ℓ = 0$ giving $w₀ = n^{3/2}$. Since both describe geometric series we get $\sum_{ℓ = 0}^{d-1} w_ℓ = 𝒪(w₀+w_{d-1}) = 𝒪(n^{3/2} + n/ε) = 𝒪(\max(n^{3/2},n/ε))$ as desired.

    The space requirement for level $ℓ$ is $|M_ℓ| = 2^{ℓ}·\log₂(1/p(n_ℓ)) + ε_ℓ 2^{ℓ} + 𝒪(\log n)$. To bound the sum $\sum_{ℓ = 0}^{d-1} |M_ℓ|$ we look at the terms in $|M_ℓ|$ separately. For the first, we will use that $\prod_{ℓ = 0}^{d-1} p(2^{d-ℓ})^{2^ℓ} = \frac{n!}{n^n}$, which can be checked by induction: $d = 0$ is clear. For $d > 0$ we have
    \begin{align*}
        &\prod_{ℓ = 0}^{d-1} p(2^{d-ℓ})^{2^ℓ}
        = p(n)·\prod_{ℓ = 1}^{d-1} p(2^{d-ℓ})^{2^ℓ}
        = p(n)·\Big(\prod_{ℓ = 1}^{d-1} p(2^{d-ℓ})^{2^{ℓ-1}}\Big)²\\
        &= p(n)·\Big(\prod_{ℓ = 0}^{d-2} p(2^{d-1-ℓ})^{2^{ℓ}}\Big)²
        \textrel{Ind.}{=} p(n) · \Big(\frac{(n/2)!}{(n/2)^{n/2}}\Big)²
        = \binom{n}{n/2} 2^{-n} · \Big(\frac{(n/2)!}{(n/2)^{n/2}}\Big)² = \frac{n!}{n^n}.
    \end{align*}
    The leading terms in $|M_ℓ|$ therefore add up to
    \[
        \sum_{ℓ = 0}^{d-1} 2^{ℓ}·\log₂(1/p(n_ℓ))
        = 
        \log₂\Big(1/\prod_{ℓ = 0}^{d-1} p(n_ℓ)^{2^{ℓ}}\Big) = \log₂\big(\frac{n^n}{n!}\big) = n \log₂ e - 𝒪(\log n).
    \]
    The overhead terms $ε_ℓ · 2^{ℓ}$ are geometrically increasing in $ℓ$. This is because
    \[ \frac{ε_{ℓ+1}·2^{ℓ+1}}{ε_ℓ · 2^{ℓ}} = \frac{\min(1,ε·n_{ℓ+1}^{3/4})}{\min(1,ε·n_ℓ^{3/4})}·2 = \frac{\min(1,ε·(2^{d-ℓ-1})^{3/4})}{\min(1,ε·(2^{d-ℓ})^{3/4})}·2 ≥ 2^{-3/4}·2 = 2^{1/4} > 1. \]
    The sum of the overhead terms is therefore dominated by the last term $ε_{d-1}2^{d-1} ≤ εn·2^{-1/4}$ giving $\sum_{ℓ = 0}^{d-1} ε_ℓ·2^{ℓ} ≤ εn/(2^{1/4}-1) < 6εn$.
    Overall we get as claimed
    \[ \textrm{space} = \sum_{ℓ = 0}^{d-1} |M_ℓ| ≤ n\log₂ e + 6εn + 𝒪(\log² n).\]
    A query needs to extract one seed from each of the $d = \log₂ n$ \consensus data structures. Since each uses uniform values for $p₁ = … = p_{n_ℓ}$, the product $\prod_{j = 1}^{i} p_j = p(n_ℓ)^i$ is easy to compute in constant time, which is sufficient to find each seed in constant time, so we get a query time of $𝒪(d)$ in total.
\end{proof}

\subsection{Bucketed \consensus-RecSplit}
\label{sec:bucketed_consensus_recsplit}
We now replace the upper layers of Monolithic \consensus-RecSplit with a minimal $k$-perfect hash function \cite{hermann2025engineering}.
This reduces query time as there are fewer layers to traverse and reduces construction time because the splits in the top layers are the most expensive ones to compute.
Most importantly, it enables efficiently handling input sizes that are not powers of two.

For any $k ∈ ℕ$ a minimal $k$-perfect hash function for a set $S$ of $n$ keys is a data structure mapping the elements of $S$ to buckets indexed with $[⌈n/k⌉]$ such that each bucket $i ∈ [⌊n/k⌋]$ receives exactly $k$ keys. If $k$ does not divide $n$ then bucket $⌈n/k⌉$ receives between $1$ and $k-1$ leftover keys.
It has long been known that the space lower bound for such a data structure is $≈ n(\log₂e- \frac 1k\log₂ \frac{k^k}{k!})$ \cite{belazzougui2009hash}, shown to be $≈ \frac{n}{2k} \log 2πk$ in \cite{KLS:PaCHash:2023}. In
\iffullversion
\cref{sec:k-perfect}
\else
the full version of this paper \cite{lehmann2025combinedArxiv}
\fi
we outline a compact minimal $k$-perfect hash function with, in expectation, $𝒪(\frac{n}{k} \log k)$ bits of space, constant query time, and linear construction time.

We call our resulting MPHF Bucketed \consensus-RecSplit. Its performance is as follows (previously stated in \cref{sec:mphf-results}).

\bucketedCRS*

\begin{proof}
    The idea is to use our monolithic construction but replace the top-most layers with a $k$-perfect hash function. More precisely, we choose $k = 2^{⌊\frac{4}{3} \log₂ 1/ε⌋} = Θ(ε^{-4/3})$ and construct a minimal $k$-perfect hash function $F$ for the input set. As discussed in
    \iffullversion
    \cref{lem:k-mphf},
    \else
    the full version,
    \fi
    the construction time of $F$ is $𝒪(n)$, the expected query time is $𝒪(1)$ and the space consumption is $𝒪(\frac{n}{k} \log k) = 𝒪(εn)$ bits, all well within the respective budgets. The resulting $⌊n/k⌋$ buckets of $k$ keys each (a power of $2$) are recursively split like in our monolithic data structure, which takes at most $\OPTMPHF+𝒪(εn)$ bits of space (the $𝒪(\log²n)$ term is dominated by $𝒪(εn)$). Since we make no assumptions on $n$ there might be up to $k-1$ leftover keys assigned to an extra bucket. For these we construct a separate MPHF $F'$ with range $\{⌊n/k⌋·k+1,…,n\}$. Since $k \sim ε^{-4/3} = ε·ε^{-7/3} ≤ εn$ pretty much any compact (not necessarily succinct) MPHF is good enough here, provided it supports queries in logarithmic time. Queries evaluate $F$ first and, if mapped to a regular bucket, evaluate a sequence of $\log₂ k$ splits, or, if mapped to the extra bucket, evaluate $F'$. This takes $𝒪(\log k) = 𝒪(\log 1/ε)$ time in the worst case.
    
    Concerning construction time, note that the top-most level in use has nodes of size $n_ℓ = k ≤ ε^{-4/3}$ and the definition of $ε_ℓ$ we made simplifies to $ε_ℓ = n_ℓ^{3/4}·ε$ without the “$\min(1,·)$”. By an argument in \cref{thm:improved-recsplit-monolithic}, the construction time is then dominated by the bottom layer and bounded by $𝒪(n/ε)$ as claimed.
\end{proof}

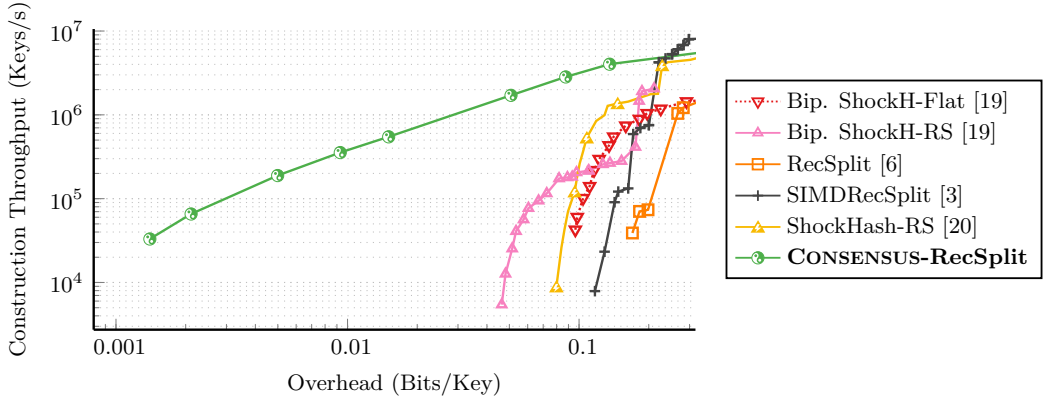
\begin{figure}[t]
  \centering
      \begin{tikzpicture}
        \begin{axis}[
            width=8cm,
            height=4cm,
            xlabel={Overhead (Bits/Key)},
            ylabel={Construction Throughput (Keys/s)},
            xmode=log,
            ymode=log,
            xmax=0.255,
            xmin=0.001,
            log x ticks with fixed point,
            legend style={at={(1.05,0.50)},anchor=west},
            legend columns=1,
          ]
          \addplot[mark=flippedTriangle,color=colorBipartiteShockHashFlat,densely dotted] coordinates { (0.09637,42552) (0.09825,60275.1) (0.1041,101340) (0.1111,140647) (0.11608,218323) (0.12213,296520) (0.13475,432135) (0.14105,548884) (0.15916,742159) (0.18205,892045) (0.19699,1.0443e+06) (0.22537,1.18029e+06) (0.29117,1.44396e+06) (0.4089,1.55902e+06) };
          \addlegendentry{Bip. ShockH-Flat \cite{lehmann2024shockhash2}}
          \addplot[mark=triangle,color=colorBipartiteShockHash,solid] coordinates { (0.04635,5443.92) (0.04813,12673.7) (0.05143,25335.9) (0.05344,40862.2) (0.05761,56383.4) (0.06046,77013.2) (0.06686,94480) (0.07252,115402) (0.08197,173902) (0.08907,179131) (0.09342,181749) (0.09751,205431) (0.1098,214964) (0.11022,215494) (0.12677,256675) (0.13598,266030) (0.15253,280599) (0.1752,413032) (0.18178,1.46064e+06) (0.18696,1.9069e+06) (0.2108,2.05389e+06) };
          \addlegendentry{Bip. ShockH-RS \cite{lehmann2024shockhash2}}
          \addplot[mark=square,color=colorRecSplit,solid] coordinates { (0.16999,38960.3) (0.18251,70474.6) (0.19856,73704.7) (0.2669,1.04029e+06) (0.28189,1.21714e+06) (0.33348,1.45575e+06) (0.34831,1.85694e+06) (0.4198,2.3092e+06) (0.43527,3.11265e+06) (0.50207,3.70178e+06) (0.5068,4.11421e+06) (0.572,5.00952e+06) (0.64938,5.42446e+06) (0.77291,5.98444e+06) (0.85055,6.65557e+06) (1.48514,6.77736e+06) (1.80332,7.227e+06) };
          \addlegendentry{RecSplit \cite{esposito2020recsplit}}
          \addplot[mark=+,color=colorSimdRecSplit,solid] coordinates { (0.11709,7888.2) (0.12888,23305.1) (0.14239,90695.8) (0.1477,121433) (0.16304,132587) (0.17132,589470) (0.18403,695696) (0.20014,754273) (0.22108,4.23263e+06) (0.23627,4.70987e+06) (0.2521,5.27232e+06) (0.26698,6.02882e+06) (0.28262,6.79948e+06) (0.29791,7.97194e+06) (0.3473,8.08407e+06) (0.36251,9.7144e+06) (0.43533,1.17343e+07) (0.50242,1.27584e+07) (0.57024,1.37438e+07) (0.62936,1.39451e+07) (0.65447,1.47341e+07) (0.70951,1.49098e+07) (1.01491,1.6239e+07) (1.13064,1.755e+07) (1.3997,1.8305e+07) (1.51757,1.83184e+07) };
          \addlegendentry{SIMDRecSplit \cite{bez2023high}}
          \addplot[mark=shockhash,color=colorShockHash,solid,mark repeat*=4] coordinates { (0.07982,8904.89) (0.08393,26133.4) (0.08648,42613.2) (0.0896,70371) (0.09574,120974) (0.0979,176158) (0.10063,263134) (0.10389,358873) (0.10794,532598) (0.11858,845594) (0.1288,992142) (0.13292,1.2811e+06) (0.14675,1.36687e+06) (0.16292,1.43808e+06) (0.20372,1.76473e+06) (0.21927,1.78272e+06) (0.22848,3.88078e+06) (0.24247,4.2508e+06) (0.30307,4.53535e+06) (0.34907,5.02892e+06) (0.36406,5.64207e+06) (0.43072,6.16599e+06) (0.50667,6.39386e+06) (0.58887,6.51254e+06) (0.67878,6.93529e+06) };
          \addlegendentry{ShockHash-RS \cite{lehmann2023shockhash}}
          \addplot[mark=consensus,color=colorConsensus] coordinates { (0.0014,33063.8) (0.00211,65692.6) (0.00499,188935) (0.0093,355245) (0.01505,547369) (0.05081,1.71113e+06) (0.08757,2.83583e+06) (0.13551,4.02609e+06) (0.6018,6.80689e+06) };
          \addlegendentry{\textbf{\consensus-RecSplit}}
        \end{axis}
    \end{tikzpicture}
  \caption{
      Experimental evaluation of very space-efficient MPHF algorithms, showing the trade-off between space consumption and construction throughput.
      Because \consensus-RecSplit is focused on space consumption, we only include approaches that achieve below 1.6 bits per key.
  }
  \label{fig:paretoMphf}
\end{figure}

\subsection{Experiments}
To demonstrate that \consensus-RecSplit is viable in practice, we give an implementation and compare it to other MPHF constructions from the literature.%
\footnote{We run single-threaded experiments on an Intel i7 11700 processor with a base clock speed of 2.5\,GHz.
The sizes of the L1 and L2 data caches are 48\,KiB and 512\,KiB per core, and the L3 cache has a size of 16\,MiB.
The page size is 4\,KiB.
We use the GNU C++ compiler version 11.2.0 with optimization flags \texttt{-O3 -march=native}.
Like previous papers \cite{lehmann2025modern,lehmann2024shockhash2,lehmann2023shockhash,hermann2024phobic,bez2023high}, we use 100 million random string keys of uniform random length $[10..50]$ as input.
The source code is available on GitHub \cite{sourceCodeMphfExperiments,sourceCodeConsensusRecSplit}.
}
We only include the most space-efficient previous approaches with space consumption below 1.6 bits per key and refer to \cite{lehmann2024fast,lehmann2025modern} for a detailed comparison of state-of-the-art perfect hashing approaches with larger space consumption.
Our implementation departs from the algorithmic description in \cref{sec:bucketed_consensus_recsplit} by using the threshold-based $k$-perfect hash function \cite{lehmann2024shockhash2,hermann2025engineering}%
\iffullversion
\ instead of the one outlined in \cref{sec:k-perfect}.
\else
.
\fi
We postpone further tuning and parallelisation efforts to a future paper.

\Cref{fig:paretoMphf} shows the trade-off between space consumption and construction time of \consensus-RecSplit.
In
\iffullversion
\cref{sec:experimentsTable},
\else
the full version \cite{lehmann2025combinedArxiv}
\fi
we also give a table.
The performance of all previous approaches degrades abruptly for smaller overheads because of their exponential dependency on $1/\varepsilon$ while \consensus-RecSplit shows an approximately linear dependence on $1/ε$ as predicted in \cref{thm:improved-recsplit-bucketed}.
The space lower bound for MPHFs is $\log_2 e \approx 1.443$ bits per key.
The previously most space-efficient approach, bipartite ShockHash-RS \cite{lehmann2024shockhash2}, achieves 1.489 bits per key.
In about 15\% of the construction time, \consensus-RecSplit achieves a space consumption of just 1.444 bits per key.
At a space consumption of 1.489 bits per key, \consensus-RecSplit is more than 350 times faster to construct than the previous state of the art.

Queries take 150--250\,ns depending on $k$ and therefore the depth of the splitting tree %
\iffullversion
(see \cref{sec:experimentsTable})%
\else
(see full version)%
\fi
, so their performance is not too far from compact configurations of RecSplit~(100\,ns) and bipartite ShockHash-RS~(125\,ns).
Future work can improve on the query performance by storing the seeds needed for each query closer to each other.
We can adapt bucketed \consensus-RecSplit by constructing one \consensus data structure storing seeds bucket by bucket instead of an independent data structure for each layer.
Preliminary experiments show up to 30\% faster queries with moderately slower construction at the same space consumption.
Splittings higher up in the tree not only have a smaller success probability, but are also more costly to test.
We are therefore incentivised to select a different space overhead parameter $\varepsilon$ for each splitting within the same \consensus data structure.
We leave this generalisation of \consensus to non-uniform $\varepsilon$ for future work.

\section{Conclusion}\label{sec:conclusion}

This paper concerns data structures that store sequences of successful seeds, i.e.\ values known to influence pseudo-random processes in desirable ways. Normally these seeds are first found independently of each other using trial and error and then encoded separately. In this paper we show that, perhaps surprisingly, this wastes $Ω(1)$ bits per seed compared to a lower bound.

We present the \consensus approach, which \textbf{co}mbi\textbf{n}es \textbf{s}earch and \textbf{en}coding of \textbf{su}ccessful \textbf{s}eeds and reduces the space overhead to $𝒪(ε)$ bits at the price of increasing the required work by a factor of $𝒪(1/ε)$.

To demonstrate the merits of our ideas in practice, we apply it to the construction of minimal perfect hash functions.
\consensus-RecSplit is the first MPHF to achieve a construction time that only depends linearly on the inverse space overhead.
Despite using the same recursive splitting idea as \cite{esposito2020recsplit}, the space overhead is, for the same construction throughput, up to two orders of magnitude smaller in practice.

\subparagraph{Future Work.}
Given our promising results, we intend to continue working on \consensus-RecSplit, in particular looking into parallelisation and improving query times.
We also believe that \consensus can be applied to other randomised data structures, including but not limited to other perfect or $k$-perfect hash functions, retrieval data structures and AMQ-filters.
We give a first result on $k$-perfect hash functions using \consensus in \cite{hermann2025engineering}.

\bibliography{paper}

\iffullversion
\clearpage
\appendix
\crefalias{section}{appendix}

\section{Lower Bounds}
\label{sec:lower-bounds-analysis}

We now prove that the values defined in \cref{eq:lower-bounds} indeed constitute  lower bounds as the naming suggests.
\begin{lemma}[Information Theoretic Lower Bounds]\ \\
    \label{lem:lower-bounds}
    Solving the SSEP requires $𝔼[|M|] ≥ M^\OPT$ and $𝔼[T_i] ≥ T_i^\OPT$ for all $i ∈ [n]$.
\end{lemma}
\begin{proof}
    Any outcome $m$ of the random variable $M$ encodes a sequence $s₁,…,sₙ$ of seeds, which are jointly successful with probability $p₁·…·pₙ$, hence $\Pr[M = m] ≤ p₁·…·pₙ$. We can bound the expected size of $M$ using its entropy $H(M)$ as follows\footnote{$𝔼[|M|] ≥ H(M)$ assumes that $M$ properly encodes its content by being, e.g. a prefix free code.}
    \begin{align*}
        𝔼[|M|]
        &≥ H(M)
        = \sum_{\mathclap{m ∈ \{0,1\}^*}} \Pr[M = m]\log₂\Big(\frac{1}{\Pr[M=m]}\Big)\\
        &≥ \sum_{\mathclap{m ∈ \{0,1\}^*}} \Pr[M = m]\log₂\Big(\frac{1}{p₁·…·pₙ}\Big)
        = \sum_{i = 1}^n \log₂ 1/p_i = M^\OPT.
    \end{align*}
    Concerning $T_i$, any correct approach must inspect random variables from $\vB{i}$ until at least the first success. In expectation this takes $1/p_i$ trials and $𝔼[T_i] ≥ 1/p_i = T_i^\OPT$ follows.
\end{proof}

\section{Properties of the Ad-Hoc Approaches}
\label{s:adhoc-proofs}

For completeness, we prove the claims we made about  the performance of the approaches MIN and UNI considered in the introduction.
\def\seeApp{}
\MINperf*
\begin{proof}[Proof of \cref{lem:ad-hoc:min}]
    It is clear that $T_i \sim \Geom(p_i)$ with $𝔼[T_i] = 1/p_i = T_i^\OPT$.
    The entropy of $X \sim \Geom(p)$ can be bounded as follows using $-\ln(1-p) ≥ p$ for $p ∈ (0,1)$.
    \[
        H(X) = \frac{-p\log₂p-(1-p)\log₂(1-p)}{p}
        ≥ -\log₂ p + (1-p)\log₂e ≥ \log₂ 1/p + 1-p.
    \]
    Since $M$ encodes a sequence of independent $S₁,…,Sₙ$ with $S_i \sim \Geom(p_i)$ we have
    \[
        𝔼[|M|] ≥ H(M) = \sum_{i = 1}^n H(S_i) ≥ \sum_{i = 1}^n \log 1/p_i + \sum_{i = 1}^n 1-p_i = M^\OPT + \sum_{i = 1}^{n} 1-p_i.\qedhere
    \]
\end{proof}

\UNIperf*
\begin{proof}[Proof of \cref{lem:ad-hoc:uni}]
    Using $\frac{x \ln x}{x-1} ≤ 1$ for $x ∈ (0,1)$ we can upper bound the entropy of $X \sim \Geom(p)$ as
    \[
        H(X) = \frac{-p\log₂p-(1-p)\log₂(1-p)}{p}
        ≤ -\log₂ p + \log₂e,
    \]
    which we can apply to $S_* \sim \Geom(p₁·…·pₙ)$ to get, using a suitable encoding $M$ of $S_*$:
    \[
        𝔼[|M|] = H(S_*) + 𝒪(1) = -\log₂(p₁·…·pₙ) + 𝒪(1) = M^\OPT + 𝒪(1).
    \]
    The number of options that need to be considered is $𝔼[S_*]$, where considering an option requires inspecting a random variable from at least one Bernoulli process. Hence
    \[ 𝔼[\sum_{i = 1}^n T_i] ≥ 𝔼[S_*] = \prod_{i = 1}^n \frac 1{p_i}.\qedhere\]
\end{proof}

\section{On Interactions Between Seeds}
\label{sec:SSSEP}
In some data structures, including RecSplit, on which our \cref{thm:improved-recsplit-monolithic,thm:improved-recsplit-bucketed} are based, seeds in a sequence $S₁,…,S_n$ need to be determined one after the other and the choice of $S₁,…,S_i$ may influence which choices for $S_{i+1}$ are successful. We sidestep the issue in \cref{thm:improved-recsplit-monolithic} by solving a sequence of SSEP instances, but alternatively we could have considered the following generalisation of the SSEP.

\begin{definition}[The Sequential Seed Search and Encoding Problem (SSSEP)]
    Let $n ∈ ℕ$ and $p₁,…,pₙ ∈ (0,1]$. Assume for each $i ∈ [n]$, each $(s₁,…,s_i) ∈ ℕ₀^{i}$ and each $j ∈ ℕ₀$ there is a random variable $B^{(s₁,…,s_{i-1})}_{j} \sim \Ber(p_i)$ such that any subset of these variables arising from a set of pairwise distinct $(i,j)$ is independent. The task is to craft a bitstring $M$ that encodes for each $i ∈ [n]$ a seed $S_i$ such that $\B{S₁,…,S_{i-1}}_{S_i} = 1$.
\end{definition}
More plainly, the SSSEP acknowledges that whether or not a seed is successful in position $i$ may depend on the choice of previous seeds, but insists that a fresh seed $s_i$, i.e.\ one not previously inspected at index $i$, yields a fresh chance for success.

It is not hard to check that our \consensus algorithms are applicable in the generalised setting: Simplified \consensus never tries the same seed in the same location and while full \consensus might do so, we treat this as a failure case in our analysis. Hence:
\begin{corollary}
    \label{cor:sssep}
    The algorithms for the SSEP analysed in this paper offer the same guarantees when applied to the SSSEP.
\end{corollary}

\section{Notes on \consensus with \texorpdfstring{$ε > 1$}{ε > 1}}
\label{sec:large-eps}

We now provide some technical considerations relating to the claim made in \cref{rem:large-eps}. The issue is that the bound $q_i = Ω(ε)$ as stated in \cref{lem:simplified-consensus} clearly fails for $ε = ω(1)$. However, a simpler argument can be used to obtain a different bound on $q_i$ when $ε ≥ 2$ (for $ε ∈ (1,2)$ the old argument still works). The assumption from \cref{thm:simplified-consensus} guarantees $k_i ≥ 2^{ε-1}/p_i$ for all $i ∈ [n]$ and the following induction shows $q_i ≥ 1-\exp(-2^{ε-2})$ and thus also $q_i ≥ \frac{1}{2}$. As before, the induction works backwards from $q_{n+1} = 1$.
\begin{align*}
    q_i &= 1-(1-p_iq_{i+1})^{k_i}
    ≥ 1-\exp(-p_i q_{i+1} k_i)
    ≥ 1-\exp(-q_{i+1} 2^{ε-1})
    \textrel{Ind.}{≥} 1-\exp(-2^{ε-2}).
\end{align*}
In simplified \consensus we get $T_i \sim \Geom(p_i q_{i+1})$ as before and using $\frac{1}{1-x} = 1+𝒪(x)$ for $x ∈ [0,\frac 12]$ we get
\[
    𝔼[T_i]
    = \frac{1}{p_i q_{i+1}}
    = \frac{T_i^\OPT}{q_{i+1}}
    ≤ \frac{T_i^\OPT}{1-\exp(-2^{ε-2})}
    ≤ T_i^\OPT (1+\exp(-Ω(2^{ε})).
\]
The additional term $n^{-c}$ in the statement of \cref{rem:large-eps} relates to the conditioning on $E$ that is required in our analysis of full \consensus.

\section{Compact \texorpdfstring{$k$}{k}-Perfect Hashing}
\label{sec:k-perfect}

The idea of $k$-perfect hashing has been around for a long time \cite{belazzougui2009hash,gonnet1988external,larson1985external}, mostly in the context of non-minimal external memory hash tables.
However, we could not find a description of a minimal $k$-perfect hash function data structure combined with compactness guarantees.
We outline one possibility here, based on a specialization of PaCHash \cite{KLS:PaCHash:2023}.
Note that our implementation of \consensus-RecSplit uses a more sophisticated minimal $k$-perfect hash function with better constant factors. We intend to elaborate on the details in a follow-up paper.

\begin{lemma}
    \label{lem:k-mphf}
    Given $k ∈ ℕ$, there exists a minimal $k$-perfect hash function with expected $𝒪(1)$-time queries, space $𝒪(\frac{n}{k}\log k)$ and construction time $𝒪(n)$.
\end{lemma}

\begin{proof}[Proof sketch.]
    Let $S ⊆ 𝒰$ be an input set of size $n$ and $h : 𝒰 → [n]$ a random hash function. Let $(x₁,…,xₙ)$ contain the elements of $S$ sorted by $h(x)$ with ties broken arbitrarily. Our minimal $k$-perfect hash function $F$ will be $F(x_i) = ⌈i/k⌉$. To represent it, we store \emph{threshold} values $t_i = h(x_{k·i})$ for $i = 1,…,⌊n/k⌋$ in a succinct predecessor query data structure $D$. The data structure can achieve expected constant time queries in this setting and space $𝒪(\log \binom{n}{n/k}) = 𝒪(\frac{n}{k} \log k)$ bits \cite{KLS:PaCHash:2023}. A minimal $k$-perfect hashing query for $x ∈ S$ considers the rank $r$ of $h(x)$ in $D$. In most cases we can conclude $F(x) = r+1$. Unfortunately, the answer is ambiguous if $h(x)$ coincides with one or several thresholds $t_{r} = … = t_{r+ℓ}$. For such keys we store a $⌈\log₂(ℓ+2)⌉$ bit value in a retrieval data structure that disambiguates between the possible values $r,…,r+ℓ+1$ for $F(x)$. It is not hard to see that the expected additional space this takes is a lower order term, provided that a compact variable-length retrieval data structure is used, e.g.\ \cite{belazzougui2013compressed}. We omit the details.
\end{proof}

\section{Additional Experimental Data}
\label{sec:experimentsTable}
In \cref{tab:overviewTable} we give a selection of data points from the Pareto fronts in \cref{fig:paretoMphf}.

\begin{table}[t]
  \caption{%
      Comparison of \consensus with other approaches from the literature.
  }
  \label{tab:overviewTable}
  \centering

\addtolength\tabcolsep{-0.5pt}
\begin{centering}
\begin{tabular}[t]{l rrr}
    \toprule
    Technique     & Bits/key    & Construction (ns/key) & Query (ns) \\ \midrule

         Bip. ShockHash-RS, $\ell$=$96$, $b$=$2000$ & 1.500 &  17\,736 & 148 \\
        Bip. ShockHash-RS, $\ell$=$128$, $b$=$2000$ & 1.489 & 183\,691 & 133 \\ \midrule
                   Bip. ShockHash-Flat, $\ell$=$96$ & 1.554 &   7\,110 &  84 \\
                  Bip. ShockHash-Flat, $\ell$=$106$ & 1.541 &  16\,591 &  78 \\ \midrule
                   RecSplit, $\ell$=$8$, $b$=$2000$ & 1.710 &      961 & 156 \\
                  RecSplit, $\ell$=$12$, $b$=$2000$ & 1.613 &  25\,667 & 139 \\ \midrule
              SIMDRecSplit, $\ell$=$12$, $b$=$2000$ & 1.614 &   1\,696 & 151 \\
              SIMDRecSplit, $\ell$=$16$, $b$=$2000$ & 1.560 & 126\,772 & 133 \\ \midrule
         \consensus, $k$=$256$, $\varepsilon$=$0.1$ & 1.578 &      248 & 140 \\
        \consensus, $k$=$512$, $\varepsilon$=$0.06$ & 1.530 &      353 & 151 \\
        \consensus, $k$=$512$, $\varepsilon$=$0.03$ & 1.494 &      584 & 150 \\
     \consensus, $k$=$32768$, $\varepsilon$=$0.006$ & 1.452 &   2\,815 & 217 \\
    \consensus, $k$=$32768$, $\varepsilon$=$0.0005$ & 1.444 &  30\,245 & 199 \\
    \bottomrule
\end{tabular}
\end{centering}

\end{table}

\fi

\end{document}